\documentclass[11pt]{article}

\usepackage{amsthm,amssymb,amsmath,enumerate}

\usepackage[colorlinks=true]{hyperref}

\hypersetup{
  linkcolor=[rgb]{0.3,0.3,0.6},
  citecolor=[rgb]{0.2, 0.6, 0.2},
  urlcolor=[rgb]{0.6, 0.2, 0.2}
}

\usepackage{palatino}
\usepackage{setspace}
\usepackage{algorithm}
\usepackage[noend]{algpseudocode}
\usepackage{tikz,tikz-qtree}
\usepackage{fullpage}

\newcommand{\un}{\mathsf{u}}
\newcommand{\sen}{\mathsf{s}}
\newcommand{\bsen}{\mathsf{bs}}
\newcommand{\depth}{\mathsf{D}}
\newcommand{\cc}{\mathsf{cc}}
\newcommand{\res}{\mathsf{Res}}
\newcommand{\fmux}{\mathsf{MUX}}
\newcommand{\fplu}{\mathsf{PLU}}
\newcommand{\fmaj}{\mathsf{MAJ}}
\newcommand{\mmux}{\mathsf{mMUX}}
\newcommand{\fand}{\mathsf{AND}}
\newcommand{\OR}{\mathsf{OR}}
\newcommand{\size}{\mathsf{size}}
\newcommand{\rand}{\mathsf{R}}
\newcommand{\qua}{\mathsf{Q}}
\newcommand{\ket}[1]{|#1\rangle}
\newcommand{\bnot}[1]{\overline{#1}}

\newcommand{\zo}{\{0, 1\}}
\newcommand{\zon}{\zo^n}
\newcommand{\zuo}{\{0, \un, 1\}}
\newcommand{\zuon}{\zuo^n}
\newcommand{\cbra}[1]{\left\{#1\right\}}
\newcommand{\uf}{\widetilde{f}}
\newcommand{\A}{\mathcal{A}}

\newtheorem{theorem}{Theorem}[section]
\newtheorem{definition}[theorem]{Definition}
\newtheorem{observation}[theorem]{Observation}
\newtheorem{remark}[theorem]{Remark}

\newtheorem{lemma}[theorem]{Lemma}
\newtheorem{claim}[theorem]{Claim}

\begin{document}

\title{Sensitivity and Query Complexity under Uncertainty}
\author{
Deepu Benson\thanks{Institute of Science \& Tech., Chinmaya Vishwa Vidyapeeth, Kerala, India. Email: {\tt bensondeepu@gmail.com}. Part of the work was done while the author was a postdoctoral fellow at Indian Institute of Technology Gandhinagar.} \and 
Balagopal Komarath\thanks{Department of Computer Science and Engineering, IIT Gandhinagar, Gujarat, India, Email:~{\tt bkomarath@rbgo.in}} \and 
Nikhil Mande \thanks{Department of Computer Science, University of Liverpool, UK, Email:~{\tt nikhil.mande@liverpool.ac.uk}} \and
Sai Soumya Nalli \thanks{Microsoft Research, Bangalore, India, Email:~{\tt saisoumya7208@gmail.com}. Part of the work was done while the author was an undergraduate student at Indian Institute of Technology Madras.} \and 
Jayalal Sarma \thanks{Department of Computer Science \& Engineering, IIT Madras, Chennai, India, Email:{\tt jayalal@cse.iitm.ac.in}} \and 
Karteek Sreenivasaiah \thanks{Department of Computer Science, University of Liverpool, UK, Email: {\tt karteek.sreenivasaiah@liverpool.ac.uk}}
}

\maketitle

\begin{abstract}
In this paper, we study the query complexity of Boolean functions in the presence of uncertainty, motivated by parallel computation with an unlimited number of processors where inputs are allowed to be unknown. We allow each query to produce three results: zero, one, or unknown. The output could also be: zero, one, or unknown, with the constraint that we should output ``unknown'' only when we cannot determine the answer from the revealed input bits. Such an extension of a Boolean function is called its \emph{hazard-free} extension.

\begin{itemize}
    \item  We prove an analogue of Huang's celebrated sensitivity theorem [Annals of Mathematics, 2019] in our model of query complexity with uncertainty.
    \item We show that the deterministic query complexity of the hazard-free extension of a Boolean function is at most quadratic in its randomized query complexity and quartic in its quantum query complexity, improving upon the best-known bounds in the Boolean world.
    \item We exhibit an exponential gap between the smallest depth (size) of decision trees computing a Boolean function, and those computing its hazard-free extension.
    \item We present general methods to convert decision trees for Boolean functions to those for their hazard-free counterparts, and show optimality of this construction. We also parameterize this result by the maximum number of unknown values in the input.
    \item We show lower bounds on size complexity of decision trees for hazard-free extensions of Boolean functions in terms of the number of prime implicants and prime implicates of the underlying Boolean function.
\end{itemize}

\end{abstract}

\clearpage
\pagenumbering{arabic}

\tableofcontents

\section{Introduction}
  
A single-tape Turing Machine needs at least $n$ steps in the worst case to compute any $n$-bit function that depends on all of its inputs. One way to achieve faster computation is to use multiple processors in parallel. Parallel computation is modeled using \emph{Parallel Random Access Machines} (PRAM), originally defined by Fortune and Wyllie \cite{FW78}, in which there are multiple processors and multiple memory cells shared among all processors. In a single time step, each processor can read a fixed number of memory cells, execute one instruction, and write to a fixed number of memory cells. In the real world, access to shared memory has to be mediated using some mechanism to avoid read-write conflicts. A popular mechanism to achieve this is called \emph{Concurrent-Read Exclusive-Write} (CREW) in which concurrent reads of the same memory cell are allowed, but each cell can only be written to by at most one processor in a single time step. An algorithm that violates this restriction is invalid.

A fundamental problem in such a model of computation is to determine the number of processors and the amount of time required for computing Boolean functions. A Boolean function $f: \{0, 1\}^n \mapsto \{0, 1\}$ is predetermined\footnote{The number of input bits is fixed, making this model non-uniform.} and the input bits are presented in shared memory locations. The processors have to write the output bit $f(x)$ to a designated shared memory location. For example, consider the Boolean disjunction (logical OR) of all input bits which outputs 1 if and only if there is at least one $1$ among the inputs. There is a simple divide-and-conquer algorithm to compute the OR of $n$ input bits in $O(\log n)$ time using $n$ processors that exploits the fact that OR is associative and distributive. This is essentially a CREW-PRAM algorithm to compute the OR of $n$ bits in $O(\log(n))$-time. Note that each of the processors in the above algorithm computes a trivial function at each step namely the OR of two bits. Can we do better if we are allowed to use more complex computations at each step?

Cook and Dwork \cite{CD82}, and Cook, Dwork, and Reischuk \cite{CDR86} answer this question by showing that any CREW-PRAM algorithm that computes the logical OR of $n$ bits needs $\Omega(\log(n))$-time, irrespective of the functions computed by the processors at each step. The reason their lower bound is independent of the functions allowed at each processor is because their lower bound really applies to the number \emph{accesses} made to the shared memory. If we only care about analyzing the number of memory accesses by an algorithm running on an all-powerful processor, a neat way to think of the computation at each processor is to model it as a two-player interactive game: a \emph{querier} who is all-powerful and an \emph{oracle}. They want to compute a Boolean function $f : \{0, 1\}^n \to  \{0, 1\}$ known beforehand, on an input $x\in \zo^n$. However, the input $x$ is only known to the oracle. The only interaction allowed is where the querier asks a query $i\in [n]$, and the oracle can reply with $x_i$. The \emph{query complexity} of $f$ is defined as the the maximum number of queries required to find the answer, where the maximum is taken over all $x$. The querier's aim is to minimize the number of queries required to determine the value of the function in the worst-case. This is exactly the computation model studied in an exciting area of computer science simply called `query complexity'.

The technique used in \cite{CD82} and \cite{CDR86} involves defining a measure that is equivalent to what is now commonly known as \emph{sensitivity} of a Boolean function, denoted $\sen(f)$. More precisely, they show that the time needed by CREW-PRAM algorithms, irrespective of the instruction set, to compute a Boolean function is asymptotically lower bounded by the logarithm of the sensitivity of the function.

The question whether a function can be computed in time at most the logarithm of the function's sensitivity, thereby characterizing the CREW-PRAM time complexity in terms of sensitivity, remained open. Nisan \cite{Nisan91} introduced a generalization of sensitivity called \emph{block-sensitivity}, denoted $\bsen(f)$, which is lower bounded by sensitivity, and proved that CREW-PRAM time complexity is asymptotically the same as the logarithm of the block sensitivity of a function. Nisan further related \emph{decision tree depth complexity}, denoted $\depth(f)$ to block-sensitivity providing another characterization for CREW-PRAM time complexity as the logarithm of decision tree depth. However, the block-sensitivity of a function can be potentially much higher than its sensitivity, as shown by Rubinstein~\cite{Rubinstein} with an explicit function that witnesses a quadratic gap between these two measures. After research spanning almost three decades which saw extensive study of relationships between the above parameters and various other parameters (see, for example,~\cite{ABKRT21} and the references therein), such as certificate complexity (denoted $\cc(f)$) and degree (denoted $\mathsf{deg}(f)$), Huang \cite{Hua19}, in a breakthrough result, proved that sensitivity and block-sensitivity are polynomially related. This is the celebrated sensitivity theorem:
\begin{theorem}[Sensitivity theorem for Boolean functions \cite{Hua19}]
    For all Boolean functions $f : \zo^n \to \zo$, $\bsen(f) = O(\sen(f)^4)$.
\end{theorem}
By earlier results~\cite{Nisan91, NS94, BBC+01}, Huang's result implies that the combinatorial and analytic measures of sensitivity, block sensitivity, certificate complexity, degree, approximate degree, deterministic/randomized/quantum query complexity are all polynomially equivalent for Boolean functions. This strong connection makes the study of query complexity essentially equivalent to studying the time complexity of CREW-PRAMs. 
Thus henceforth, we shall predominantly use terminology from query complexity, and also express our results in the context of query complexity.

In this work, we initiate a systematic study to understand the effect of allowing \emph{uncertainty} among the inputs in the setting of CREW-PRAM. Allowing uncertainty in this setting is easier to understand in the equivalent query model: When an input bit is queried by the querier, the oracle can reply with ``uncertain''. If it is possible for the function value to be determined in the presence of such uncertainties, then we would like the querier to output such a value. In what follows, we make the setting more formal.

To model uncertainty we use a classic three-valued logic, namely Kleene's strong logic of indeterminacy  \cite{Kleene52}, usually denoted \emph{K3}. The logic K3 has three truth values $0$, $1$, and $\un$. The behaviour of the third value $\un$ with respect to the basic Boolean primitives --- conjunction ($\wedge$), disjunction ($\vee$), and negation ($\neg$) --- are given in Table~\ref{table:K3}.
\begin{table}
\centering
\begin{tabular}{l|lll}
$\wedge$ & 0 & $\un$ & 1 \\
\hline
0 & 0 & 0 & 0 \\
$\un$ & 0 & $\un$ & $\un$ \\
1 & 0 & $\un$ & 1
\end{tabular}
\quad
\begin{tabular}{l|lll}
$\vee$  & 0 & $\un$ & 1 \\
\hline
0 & 0 & $\un$ & 1 \\
$\un$ & $\un$ & $\un$ & 1 \\
1 & 1 & 1 & 1
\end{tabular}
\quad
\begin{tabular}{l|lll}
$\neg$ & 0 & $\un$ & 1 \\
\hline
    & 1 & $\un$ & 0
\end{tabular}

\caption{Kleene's three valued logic `K3'}
\label{table:K3}
\end{table}

The logic K3 has been used in several other contexts where there is a need to represent and work with \emph{unknowns} and hence has found several important wide-ranging applications in computer science. For instance, in relational database theory, SQL implements K3 and uses $\un$ to represent a NULL value \cite{Meyden98}. Perhaps the oldest use of K3 is in modeling \emph{hazards} that occur in real-world combinational circuits. Recently there have been a series of results studying constructions and complexity of hazard-free circuits~\cite{ikenmeyer2019complexity, IKS, Jukna21}. Here $\un$ was used to represent \emph{metastability}, or an \emph{unstable} voltage, that can resolve to $0$ or $1$ at a later time. 

One way to interpret the basic binary operations $\vee$, $\wedge$, and $\neg$ in K3 is as follows: for a bit $b \in \zo$, if the value of the function is guaranteed to be $b$ regardless of all $\zo$-settings to the $\un$-variables, then the output is $b$. Otherwise, the output is $\un$. This interpretation generalizes in a nice way to $n$-variate Boolean functions. In literature, this extension is typically called the \emph{hazard-free} extension of $f$ (see, for instance, \cite{ikenmeyer2019complexity}), and is an important concept studied in circuits and switching network theory since the 1950s. The interested reader can refer to \cite{ikenmeyer2019complexity}, and the references therein, for history and applications of this particular way of extending $f$ to K3. We define this extension formally below.

For a string $x \in \zuo^n$, define the \emph{resolutions} of $x$ as follows:
\[
    \res(x) := \cbra{y \in \zo^n : y_i = x_i~\forall i \in [n]~\textnormal{with}~x_i \in \zo}.
\]
That is, $\res(x)$ denotes the set of all strings in $\zo^n$ that are consistent with the $\{0,1\}$-valued bits of $x$. The hazard-free extension of a Boolean function is defined as follows:

\begin{definition}[Hazard-free Extensions]
    \label{def:hfe}
    For a Boolean function $f:\zo^n \to \zo$, we define its hazard-free extension $\uf : \zuo^n \to \zuo$ as follows. For an input $y \in \zuo^n$,
\[
\uf(y) := \begin{cases}
 	b & \textrm{if $f(y)=b$ for all $y \in \res(x)$}  \\
 	u & \textrm{ otherwise }
 	\end{cases}
\]
\end{definition}

To understand the motivation behind this definition, consider the \emph{instability partial order} defined on $\zuo$ by the relations $\un \leq 0$ and $\un \leq 1$. The elements $0$ and $1$ are incomparable. This partial order captures the intuition that $\un$ is less certain than $0$ and $1$. This partial order can be extended naturally to elements of $\zuon$ as $x \leq y$ iff $x_i \leq y_i$ for all $1 \leq i \leq n$. A function $f' : \zuon \to \zuo$ is \emph{natural} if for all $x, y \in \zuon$ such that $x \leq y$, we have $f'(x) \leq f'(y)$ and $f'(z) \in \zo$ when $z \in \zon$. Intuitively, this property says that the function cannot produce a less certain output on a more certain input and if there is no uncertainty in the input, there should be no uncertainty in the output. A natural function $f'$ extends a Boolean function $f$ if $f'(x) = f(x)$ for all $x \in \zon$. There could be many natural functions that extend a Boolean function. Consider two natural functions $f'$ and $f''$ that extend $f$. We say $f' \leq f''$ if $f'(x) \leq f''(x)$ for all $x \in \zuon$. This says that the output of $f''$ is at least as certain as the output of $f'$. An alternative definition for the hazard-free extension of a function $f$ is as follows: it is the unique function $\uf$ such that $f' \leq \uf$ for all natural functions $f'$ that extends $f$. That is, the hazard-free extension of a Boolean function is the best we could hope to compute in the presence of uncertainties in the inputs to $f$.

 We note here that even though we use the term ``hazard-free'', there is a fundamental difference between our model of computation and the ones studied in results such as \cite{ikenmeyer2019complexity, IKS, Jukna21}. For hazard-free circuits, the value $\un$ represents an unstable voltage and the gates in a circuit fundamentally are unable to detect it. That is, there is no circuit that can output $1$ when input is $\un$ and $0$ otherwise. However, in our setting, the value $\un$ is simply another symbol just like $0$ or $1$. So we can indeed detect/read a $\un$ value. The restriction that we have to compute the hazard-free extension is a semantic one in this paper whereas Boolean circuits can only compute natural functions. In other words, we are interested in query complexity of hazard-free extensions of Boolean functions per se, and we have no notion of metastability in our computation model.

There is a rich body of work in query complexity of Boolean functions that has established best-possible polynomial relationships among various models such as deterministic, randomized, and quantum models of query complexity, and their interplay with analytical measures such as block sensitivity, and certificate complexity (see, for example,~\cite{ABKRT21} and the references therein). We study these relationships in the presence of uncertainty, in particular for hazard-free extensions of Boolean functions. A main goal is to characterize query complexity (equivalently CREW-PRAM time complexity) of hazard-free extensions of Boolean functions using these parameters.
\subsection{Our Results}

In this subsection, we discuss the results presented in this paper. The organization of the paper follows the structure of this subsection.

\subsubsection{Sensitivity Theorem in the Presence of Uncertainty}

We prove the sensitivity theorem for Boolean functions in the presence of uncertainties. We define analogues of query complexity, sensitivity, block sensitivity, and certificate complexity called $\un$-query complexity (denoted $\depth_\un(f)$), $\un$-sensitivity (denoted $\sen_\un(f)$), $\un$-block sensitivity (denoted $\bsen_\un(f)$), and $\un$-certificate complexity (denoted $\cc_\un(f)$), respectively. We show that analogues of Cook, Dwork, and Reischuk's \cite{CDR86} lower bound and Nisan's \cite{Nisan91} upper bound hold in the presence of uncertainties by adapting their proofs suitably (See Appendix~\ref{app:cooknisanun}). Therefore, we can now focus on proving that these parameters are polynomially equivalent in the presence of uncertainties. Huang's proof of the sensitivity theorem for Boolean functions crucially uses the parameter called the degree of a Boolean function. It is unclear how to define an analogue of degree in our setting. However, it turns out that a more classical parameter, the maximum of prime implicant size and prime implicate size, suffices.

Our proof of the sensitivity theorem in the presence of uncertainty is much simpler and more straightforward than the proof of the sensitivity theorem in the classical setting. It raises the question of whether we can find a simpler proof of the classical sensitivity theorem by generalizing our proof to handle an arbitrary upper bound on the number of uncertain values in the input. Note that the classical setting assumes that this number is 0 and we prove the sensitivity theorem by assuming that this number is $n$, the number of inputs.

Recall that a \emph{literal} is an input variable or its negation. An \emph{implicant} (\emph{implicate}) of a Boolean function $f$ is a subset $S$ of all literals such that $f$ is $1$ (is $0$) on any input that has all literals in $S$ set to $1$ (set to $0$ respectively). A \emph{prime implicant} (\emph{prime implicate}) is an implicant (implicate) such that no proper subset of it is an implicant (implicate respectively), i.e., the implicant (implicate) is minimal (w.r.t.~set inclusion). The \emph{size} of a prime implicant or prime implicate is the size of the set. Prime implicants and prime implicates of a Boolean function are widely studied in electronic circuit design and Boolean function analysis.

\begin{theorem}[Sensitivity theorem for hazard-free extensions of Boolean functions]\label{thm:sensthm}
    Let $f : \zon \to \zo$ be a Boolean function and let $k_1$ and $k_2$ be the sizes of a largest prime implicant and prime implicate of $f$. Then, the parameters $\sen_\un(f)$, $\bsen_\un(f)$, $\depth_\un(f)$, $\cc_\un(f)$, and $\max{\{k_1, k_2\}}$ are linearly equivalent.
\end{theorem}
We note here that while Huang~\cite{Hua19} showed $\bsen(f) = O(s(f)^4)$ for all Boolean $f$, our result shows that, in the presence of uncertainty block sensitivity and sensitivity are in fact linearly related to each other.

\subsubsection{Deterministic, Quantum and Randomized Models}

Clearly, the query complexity of the hazard-free extension of a Boolean function $f$ cannot be smaller than that of $f$ itself. Can the query complexity of the hazard-free extension of a Boolean function be much more than the query complexity of the function itself? For monotone functions, we show that the answer is no. In fact, this also holds for randomized query complexity (denoted by $\rand(f)$ and $\rand_\un(f)$) and quantum query complexity (denoted by $\qua(f)$ and $\qua_\un(f)$).
\begin{lemma}\label{lem:monsame}
    Let $f : \zon \to \zo$ be a monotone Boolean function. Then we have
    \[
        \depth_\un(f) = \Theta(\depth(f)) \text{ and } \rand_\un(f) = \Theta(\rand(f)) \text{ and } \qua_\un(f) = \Theta(\qua(f)).  
    \]
\end{lemma}

A natural question to ask is whether the model we are considering is non-degenerate. In other words, do all hazard-free extensions of Boolean functions have large query complexity?
Using Lemma~\ref{lem:monsame}, we can show that there are functions that are easy to compute even in the presence of uncertainties. The following monotone variant of the $\fmux$ function (defined below) by Wegener~\cite{Wegener85} is sufficient.
\begin{definition}[\cite{Wegener85}]\label{defn: monotone indexing}
For an even integer $n > 0$, define $\mmux_n$, as follows: The function $\mmux_n$ is defined on $n + \binom{n}{n/2}$ (which is $\Theta(2^n/\sqrt{n})$) variables, where the latter $\binom{n}{n/2}$ variables are indexed by all $n$-bit strings of Hamming weight exactly $n/2$. For $(x, y) \in \zo^{n + \binom{n}{n/2}}$, define
\[
    \mmux_n(x, y) = \begin{cases}
    0 & |x| < n/2\\
    1 & |x| > n/2\\
    y_x & \textnormal{otherwise}.
    \end{cases}
\]    
\end{definition}
It is easy to see that this function is monotone and has a query complexity of $n+1$. We also exhibit a non-monotone, non-degenerate $n$-variate function such that its hazard-free extension has $O(\log n)$ query complexity in Section~\ref{app:nonmoneasy}.

For general functions, we show that uncertainty can blow-up query complexity exponentially. The Boolean function $\fmux_n : \zo^{n + 2^n} \to \zo$ defined by
$$\fmux_n(s_0, s_1, \dotsc, s_{n-1}, (x_{(b_0, \dotsc, b_{n-1})})_{b_i\in\{0,1\}}) := x_{(s_0, \dotsc, s_{n-1})}$$
is a function on $n+2^n$ inputs that depends on all its inputs and has query complexity of $n+1$. The inputs $s_i$ are called the \emph{selector} bits and the inputs $x_j$ are called \emph{data} bits. It is easy to show that any function that depends on all its $N$ input bits must have at least logarithmic (in $N$) query complexity. Therefore, $\fmux_n$ is one of the easiest functions to compute in the query complexity model. We prove that its hazard-free extension is one of the hardest functions in the query complexity model.
\begin{theorem}\label{thm:muxdepth}
    \[ 
    \depth_u(\fmux_n) = 2^n + n \text{ and } \rand_\un(\fmux_n) = \Theta(2^n) \text{ and } \qua_\un(\fmux_n) = \Theta(2^{n/2})
    \]
\end{theorem}

We also show the following relationships between deterministic, randomized, and quantum query complexities of hazard-free extensions of Boolean functions.
\begin{theorem}\label{thm: main}
    For $f : \zon \to \zo$, we have
    \[
    \depth_\un(f) = O(\rand_\un(f)^2) \text{ and } \depth_\un(f) = O(\qua_\un(f)^4).
    \]
\end{theorem}

We remark here that the deterministic-randomized relationship above is \emph{better} than the best-known cubic relationship in the Boolean world, while the quartic deterministic-quantum separation above matches the best-known separation in the Boolean world (see~\cite{ABKRT21}). The key reason we are able to obtain better relationships is because we show that sensitivity, block sensitivity, and certificate complexity are all linearly related in the presence of uncertainty (Theorem~\ref{thm:sensthm}). This is not the case in the classical Boolean setting. Regarding best possible separations: while a linear relationship between deterministic and randomized query complexities in the presence of uncertainty remains open, a quadratic deterministic-quantum separation follows from Theorem~\ref{thm:muxdepth} (or from the OR function, which has maximal deterministic query complexity, but Grover's search algorithm~\cite{Gro96} offers a quadratic quantum speedup).

It is natural to model query algorithms using \emph{decision trees}. A decision tree represents the strategy of the querier using a tree structure. The internal nodes of the tree are labeled by input variables and has two outgoing edges labeled $0$ and $1$. Computation starts at the root and queries the variable labeling the current node. Then, the outgoing edge labeled by the answer is taken to reach the next node. The leaves of the tree are labeled with $0$ or $1$ and represent the answer determined by the querier. A decision tree is said to compute a Boolean function $f$ if the querier can correctly answer the value of $f$ for every possible input by following the decision tree from root to a leaf. The \emph{depth} of the decision tree represents the worst-case time. The depth of a smallest depth decision tree that computes $f$ is called the \emph{decision tree depth complexity} of $f$, and it is the same as \emph{deterministic query complexity} of $f$.

\subsubsection{Decision Tree Size}

Decision trees have played an important role in machine learning. The size of a decision tree is an important measure as large decision trees often suffer from over-fitting. It has been long-known that functions that admit small-size decision trees are efficiently PAC learnable~\cite{EH89}. Moreover, the class of decision trees of large size is not efficiently PAC-learnable as their VC-dimension is directly proportional to their size. Various ideas to prune decision trees and reduce their size while maintaining reasonable empirical error (error on the training set) are used in practice. For a detailed treatment of the role of decision tree size in learning, the interested reader may refer to \cite[Chapter 18]{ShalevBenDavid14}. We denote the decision tree size of a function $f$ by $\size(f)$ and the decision tree size of its hazard-free extension by $\size_\un(f)$. We show that for the $\fmux$ function despite the exponential blow-up in depth from Theorem~\ref{thm:muxdepth}, the size blow-up is only polynomial.
\begin{theorem}\label{thm:muxsize}
    $2\cdot 4^n \leq \size_\un(\fmux_n) \leq 4^{n+1} - 3^n$.
\end{theorem}

In contrast, we show that there are functions for which the size blow-up is exponential. In particular, the $\fand_n$ function has linear-size Boolean decision trees. However, its hazard-free extension needs exponential-size decision trees.
\begin{theorem}\label{thm:andsize}
    $\size_\un(\fand_n) = 2^{n+1}-1$.
\end{theorem}

We also show that there are hazard-free extensions of Boolean functions that require trees of size $\Omega\left(\binom{n}{n/3}\binom{2n/3}{n/3}\right)$ (see Section~\ref{app:cm}). Notice that a ternary tree of depth $n$ can have at most $3^n$ leaves. This lower bound is only smaller than this worst-case by a polynomial factor.

We also show how to construct decision trees for hazard-free extensions of Boolean functions from a decision tree for the underlying Boolean function.
\begin{theorem}\label{thm:dttoudt}
    For any Boolean function $f$, we have $2\size(f)-1 \leq \size_\un(f) \leq 2^{\size(f)} - 1$.
\end{theorem}
The tightness of the first inequality is witnessed by the {\sf PARITY}$_n$ function and that of the second inequality is witnessed by the $\fand_n$ function. 

We also show that, in the case of hazard-free extensions too, sensitivity plays an important role in the function learning problem. The problem is as follows: We are provided a few input-output pairs and a guarantee that the function is from some family. The goal is to learn the function from as few samples as possible. It is known that a function with sensitivity $s$ is completely specified by its values on a Hamming ball of radius $2s$~\cite{GNSTW16}. We prove an analogue for hazard-free extensions of Boolean functions in Section~\ref{app:learning-sensitivity}.

\begin{theorem}\label{thm:senball}
    A hazard-free extension $f$ that has $\sen_u(f) \le s$ is specified by its values on any Hamming ball of radius $4s$ in $\zuon$.
\end{theorem}

\subsubsection{Limited Uncertainty}

We study computation in the presence of only a limited amount of uncertainty by introducing a parameter $k$ that limits the number of bits for which the oracle can respond $\un$. For metastability-containing electronic circuits, it is known that assuming only limited metastability allows constructing circuits that are significantly smaller \cite{ikenmeyer2019complexity}. We show a similar effect on decision tree size and query complexity when uncertainty is limited in Section~\ref{app:kbit}.
\begin{theorem}\label{thm:kbitsize}
    Let $T$ be a Boolean decision tree of size $s$ and depth $d$ for $f$. Then, there exists a decision tree of size at most $s^{2^{k+1} - 1}$ and depth at most $2^k \cdot d$ for $\uf$ provided that the input is guaranteed to have at most $k$ positions with value $\un$.
\end{theorem}
For settings in which $k$ is a small constant, observe that the decision tree size is polynomial in the size of the Boolean decision tree. If $k$ is considered a parameter, observe that the depth remains fixed-parameter tractable in the language of parameterized complexity theory.

\section{Preliminaries}

\begin{definition}[\textbf{$\un$-query complexity}]
    Let $f : \zon \to \zo$ be a Boolean function, and let $\uf$ be its hazard-free extension. A \emph{deterministic decision tree}, also called a \emph{deterministic query algorithm}, computing $\uf$, is a ternary tree $T$ whose leaf nodes are labeled by elements of $\zuo$, each internal node is labeled by a variable $x_i$ where $i \in [n]$ and has three outgoing edges, labeled $0$, $1$, and $\un$. On an input $x \in \zuon$, the tree's computation proceeds from the root down to a leaf as follows: from a node labeled $x_i$, we take the outgoing edge labeled by value of $x_i$ until we reach a leaf. The label of the leaf is the output of the tree $T(x)$.

    We say that $T$ computes $\uf$ if $T(x) = \uf(x)$ for all $x\in\zuon$. The deterministic $\un$-query complexity of $f$, denoted $\depth_\un(f)$, is defined as
    \[
        \depth_\un(f) := \min_{T} \textnormal{depth}(T),
    \]
    where the minimization is over all deterministic decision trees $T$ that compute $\uf$.
\end{definition}

\begin{definition}[\textbf{$\un$-sensitivity}]          \label{def:hfsens}
    Let $f : \zon \to \zo$ be a Boolean function, and let $\uf$ be its hazard-free extension. For an $x\in \zuon$, we define the $\un$-sensitivity of $f$ at $x$ as:
    \begin{equation*}
        \sen_\un(f, x) = |\{i \mid \text{$\exists y \in \zuon$ s.t. $\uf(y) \neq \uf(x)$ and $y_j \neq x_j$ at only $j=i$}\}|
    \end{equation*}
    The elements $i$ of the set are called the \emph{$\un$-sensitive bits of $x$ for $f$}. The \emph{$\un$-sensitivity of $f$}, denoted $\sen_\un(f)$, is $\max_{x\in\zuon} \sen_\un(f, x)$.
\end{definition}

\begin{definition}[\textbf{$\un$-block sensitivity}]
    Let $f : \zon \to \zo$ be a Boolean function and let $\uf$ be its hazard-free extension. For $x \in \zuon$, the $\un$-block sensitivity of $f$ at $x$ is defined as maximum $k$ such that there are disjoint subsets $B_1, B_2, \ldots B_k \subseteq [n]$ and for each $i \in[k]$, there is a $y \in \zuon$ such that $\uf(y) \neq \uf(x)$ and $y$ differs from $x$ at exactly the positions in $B_i$. Each $B_i$ is called a $\un$-sensitive block of $f$ on input $x$. The $\un$-block sensitivity of $f$, denoted $\bsen_\un(f)$, is then defined as the maximum $\un$-block sensitivity of $f$ taken over all $x$.
\end{definition}
For $b \in \zo$, we use $\bsen_{\un, b}(f)$ to denote the maximum $\un$-block sensitivity of $f$ at $x$, over all $x \in f^{-1}(b)$. 
For a string $x \in \zuon$ and a set $B \subseteq [n]$ that is a sensitive block of $\uf$ at $x$, we abuse notation and use $x^B$ to denote an arbitrary but fixed string $y \in \zuon$ that satisfies $y_{[n] \setminus B} = x_{[n] \setminus B}$ and $\uf(y) \neq \uf(x)$.

We now formally define certificate complexity for hazard-free extensions of Boolean functions. We first define a \emph{partial assignment} as follows: By a partial assignment on $n$ bits, we mean a string $p \in \cbra{0, 1, \un, *}^n$, representing partial knowledge of a string in $\zuon$, where the $*$-entries are yet to be determined. We say a string $y \in \zuon$ is \emph{consistent} with a partial assignment $p \in \cbra{0, 1, \un, *}^n$ if $y_i = p_i$ for all $i \in [n]$ with $p_i \neq *$.

\begin{definition}[\textbf{$\un$-certificate complexity}]
    Let $f : \zon \to \zo$ and $x \in \zuon$. A partial assignment $p \in \cbra{0, 1, \un, *}^n$ is called a \emph{certificate} for $\uf$ at $x$ if 
    \begin{itemize}
        \item $x$ is consistent with $p$, and
        \item $\uf(y) = \uf(x)$ for all $y$ consistent with $p$.
    \end{itemize}
    The size of this certificate is $|p| := |{\cbra{i \in [n] : p_i \neq *}}|$. The \emph{domain of $p$} is said to be $\cbra{i \in [n] : p_i \neq *}$. The \emph{certificate complexity of $\uf$ at $x \in \zuon$}, denoted $\cc_\un(f,x)$, is the minimum size of a certificate $p$ for $\uf$ at $x$. The \emph{$\un$-certificate complexity} of $\uf$, denoted $\cc_\un(f)$, is the maximum value of $\cc_{\un}(f, x)$ over all $x$.
\end{definition}
In other words, a certificate for $\uf$ at $x$ is a set of variables of $x$ that if revealed, guarantees the output of all consistent strings with the revealed variables to be equal to $\uf(x)$.

\begin{definition}[\textbf{Randomized $\un$-query complexity}]
    A randomized decision tree is a distribution over deterministic decision trees. We say a randomized decision tree computes $\uf$ with error $1/3$ if for all $x \in \zuon$, the probability that it outputs $\uf(x)$ is at least $2/3$. The depth of a randomized decision tree is the maximum depth of a deterministic decision tree in its support.
    Define the randomized $\un$-query complexity of $f$ as follows.
    \[
        \rand_\un(f) := \min_{T}\textnormal{depth}(T),
    \]
    where the minimization is over all randomized decision trees $T$ that compute $\uf$ to error at most $1/3$.
\end{definition}

We refer the reader to~\cite{BW02} for basics of quantum query complexity.

\begin{definition}[\textbf{Quantum $\un$-query Complexity}]
    A quantum query algorithm $\mathcal{A}$ for $\uf$ begins in a fixed initial state $\ket{\psi_0}$ in a finite-dimensional Hilbert space, applies a sequence of unitaries $U_0, O_x, U_1, O_x, \dots, U_T$, and performs a measurement. Here, the initial state $\ket{\psi_0}$ and the unitaries $U_0, U_1, \dots, U_T$ are independent of the input. The unitary $O_x$ represents the ``query'' operation, and does the following for each basis state: it maps $\ket{i}\ket{b}\ket{w}$ to $\ket{i}\ket{b + x_i \mod 3}\ket{w}$ for all $i \in [n]$ (here $x_i = \un$ is interpreted as $x_i = 2$, and the last register represents workspace that is not affected by the application of a query oracle).
    
    The algorithm then performs a 3-outcome measurement on a designated output qutrit and outputs the observed value.

    We say that $\mathcal{A}$ is a bounded-error quantum query algorithm computing $\uf$ if for all $x \in \zuon$ the probability that $\uf(x)$ is output is at least $2/3$. The (bounded-error) \emph{quantum $\un$-query complexity of $\uf$}, denoted by $\qua_\un(f)$, is the least number of queries required for a quantum query algorithm to compute $\uf$ with error probability at most $1/3$.
\end{definition}

\section{Sensitivity Theorem in the Presence of Uncertainty}

In this section, we present the proof of the sensitivity theorem in the presence of uncertainty. We also present a characterization of the sensitivity of hazard-free extensions in the second subsection which results in an improved relationship between sensitivity and certificate complexity in the context of hazard-free extensions.

\subsection{Relating Sensitivity, Block Sensitivity and Certificate Complexity}

We first show that sensitivity, block sensitivity, certificate complexity, and size of the largest prime implicant/prime implicate are all asymptotically the same. This is a more formal and more precise restatement of Theorem~\ref{thm:sensthm}.
\begin{theorem}
    \label{thm:sbscc}
    Let $f:\zon\to \zo$ be a Boolean function. Let $k_1$ be the size of a largest prime implicant of $f$, and $k_2$ be the size of a largest prime implicate of $f$. Then, we have:
    \begin{equation*}
        \max\{k_1, k_2\} \leq \sen_\un(f) \leq \bsen_\un(f) \leq \cc_\un(f) \leq k_1+k_2-1.
    \end{equation*}
\end{theorem}
\begin{proof}
    The inequalities $\sen_\un(f) \leq \bsen_\un(f) \leq \cc_\un(f)$ follow from definitions similar to their Boolean counterparts (see, for example,~\cite[Proposition 1]{BW02}).

    To show the first inequality, we crucially use our three-valued domain. Let $P$ be a prime implicant of $f$. Define the input $y_P\in \zuon$ to be $1$ in those positions where $P$ contains the corresponding positive literal, $0$ where $P$ has the corresponding negated literal, and $\un$ everywhere else. We claim that each index in $\cbra{i \in [n] : x_i \in P \textnormal{ or } \neg x_i \in P}$ is sensitive for $\uf$ at $y_P$. To see this, first observe that $\uf(y_P) = 1$ since $P$ being an implicant means $f$ is $1$ on all resolutions of $y_P$. Let $x_i\in P$ be a positive literal. Then, observe that if setting the $i$'th bit of $y_P$ to a $0$ does not change the value of $\uf$, then changing the $i$'th bit to $\un$ would also not change the value of $\uf$. This means $P\setminus \{x_i\}$ would also be an implicant of $f$, contradicting the fact that $P$ was a \emph{prime} implicant. The case when $x_i$ appears as a negated literal in $P$ is similar. Thus we have $\sen_\un(f) \geq k_1$. A similar argument can be made for prime implicates as well, showing that $\sen_\un(f) \geq k_2$. This proves the first inequality in the theorem.

    To prove the last inequality, let $x\in \zuon$ be any input to $\uf$. Observe that if $\uf(x)=0$, then the prover can pick an implicate (of size at most $k_2$) that has all literals set to $0$ in $x$, and reveal those values to the verifier. If $\uf(x)=1$, then the prover reveals an implicant (of size at most $k_1$) that is $1$. If $\uf(x)=\un$, then there must exist inputs $x^0,x^1\in \res(x)$ such that $f(x^0)=0$ and $f(x^1)=1$. Since both $x^0$ and $x^1$ are resolutions of $x$, it must hold that for all $i\in [n]$ where $x_i\in\{0,1\}$, $x_i^0=x_i^1=x_i$. i.e., $x^0$ and $x^1$ differ from $x$ only in positions where $x$ is $\un$. Hence, a prime implicant that is $1$ in $x^1$ must contain a $\un$ in $x$. Similarly, a prime implicate that is $0$ in $x^0$ must contain a $\un$ in $x$. Thus, there exists a prime implicant and a prime implicate of $f$ both of which contain a literal assigned $\un$ in $x$. The prover reveals the bits in such a prime implicant and prime implicate. Since every prime implicant and every prime implicate have at least one common variable, the prover reveals only $k_1+k_2-1$ values. Why should this convince the verifier? Since $x^0$, $x^1$, and $x$ coincide on positions where $x$ has $0$ or $1$, it is possible to set the $\un$ in the revealed prime implicant to values that make the function output $1$. Similarly, it is possible to make the function output $0$ by filling in the values to the $\un$ positions in the prime implicate revealed. Thus, the verifier can check that there are indeed two valid resolutions that give different outputs.
\end{proof}

\begin{remark}
    Notice that Theorem~\ref{thm:sbscc} shows that in our setting, the parameters sensitivity, block sensitivity, and certificate complexity are equivalent (up to a multiplicative constant) to the largest prime implicants/prime implicates. In the Boolean world, for certificate complexity we get a tighter characterization in terms of CNF/DNF width, which is analogous to prime implicant/prime implicant size here. On the other hand, in the Boolean world, there is a quadratic separation between sensitivity and block sensitivity~\cite{Rubinstein}.
\end{remark}

We now proceed to show that similar to their Boolean counterparts, the deterministic $\un$-query complexity is polynomially upper bounded by $\un$-sensitivity and deterministic, randomized, and quantum $\un$-query complexities are all polynomially related to each other. We do this in two parts:
\begin{enumerate}[(i)]
    \item $\depth_\un(f) \leq \cc_\un(f) \cdot \bsen_\un(f)$.
    \item $\bsen_\un(f) = O(\rand_\un(f))$, ${\bsen_\un(f)} = O(\qua_\un(f)^2)$.
\end{enumerate}

We start by showing $\depth_\un(f) = O(\cc_\un(f) \cdot \bsen_\un(f))$. Algorithm \ref{algo:UpperBoundu} is a deterministic query algorithm that achieves this bound, as shown in Theorem \ref{thm: algo upper bound}. Following this, we derive (ii) in Lemma \ref{lem: unstable randomized quantum bs}. The final relationships among the three $\un$-query complexities is presented in Theorem~\ref{thm:hazard-free-relationships}.

\begin{algorithm}[ht]
    \begin{algorithmic}[1]
        \State \textbf{Given:} Known $f : \zon \to \zo$; Query access to an unknown $x \in \zuo^n$.
        \State \textbf{Goal:} Output $\uf(x)$
        \State Initialize partial assignment $x^*\gets *^n$
        \For{$i \gets 1$ to $\max(\bsen_{\un,0}(\uf), \bsen_{\un,1}(\uf))$}\label{line:For}
            \State {$c \gets$ a minimum $\un$-certificate of $\uf$ at an arbitrary $x \in \uf^{-1}(\un)$ consistent with $x^*$}\label{line: firstCertificate1}
            \State $C \gets$ domain of $c$
            \State Query all variables in $C$
            \State Update $x^*$ with the answers from the oracle.
            \If{ $x^*$ is a $\un$-certificate of $\uf$}\label{line:Ifu}
                \State \textbf{Output} $\un$
            \EndIf
            \If{ $x^*$ is a $0$-certificate of $\uf$}\label{line:If0}
                \State \textbf{Output} 0
            \EndIf     
            \If{ $x^*$ is a $1$-certificate of $\uf$}\label{line:If1}
                \State \textbf{Output} 1
            \EndIf
       
        \EndFor

        \State \textbf{Output $\un$}\label{line:lastOutputu}

        \caption{$\un$-query algorithm}
        \label{algo:UpperBoundu}
    \end{algorithmic}
\end{algorithm}

We will need the following observations to prove correctness of the algorithm:

\begin{observation}\label{obs:existsucert}
    Let $x^*$ be a partial assignment that does not contain a $b$-certificate of $\uf$ for any $b\in \zuo$. Then there exists a $z\in \uf^{-1}(\un)$ that is consistent with $x^*$.        
\end{observation}
\begin{proof}
    Since $x^*$ does not already contain a $b$-certificate of $\uf$ for any $b\in \zuo$, this means that for each $b \in \zuo$, there exists a setting to the $*$-variables in the partial assignment to yield a $b$-input to $\uf$. Specifically, it must be the case that there exists some assignment to the undetermined $*$'s that sets an implicant of $f$ to $1$, and some assignment that sets an implicate of $f$ to $0$. Define the string $z$ to be the same as $x^*$ except with $*$'s replaced with $\un$.
\end{proof}

\begin{observation}\label{obs: unstable certificate no u}
    Let $f : \zon \to \zo$, $b \in \zo$ and $x \in \uf^{-1}(b)$. Let $c$ be a minimal certificate of $\uf$ at $x$, and let its domain be $C$. Then for all $i\in C$, we have $c_i \neq \un$.
\end{observation}
\begin{proof}
Assume towards a contradiction an index $i \in [n]$ in a minimal certificate $c$ for $\uf$ at $x$ with $c_i = \un$. By the definition of a certificate, and $\uf$, the partial assignment $c'$ obtained by removing $i$ from the domain of $c$ is such that all strings $x \in \zuo^n$ consistent with $c'$ satisfy $\uf(x) = b$. This shows that $c'$ is also a certificate for $\uf$ at $x$, contradicting the minimality of $c$.
\end{proof}

\begin{lemma}
    If Algorithm \ref{algo:UpperBoundu} reaches Line \ref{line:lastOutputu}, then every $1$-input of $\uf$, and every $0$-input of $\uf$ is inconsistent with the partial assignment $x^*$ (at Line \ref{line:lastOutputu}).
    \label{lem:firstInconsistency1}
\end{lemma}
\begin{proof}
  Let $k$ denote $\max(\bsen_{\un,0}(\uf), \bsen_{\un,1}(\uf))$, the number of iterations of the \textbf{for} loop on Line \ref{line:For}. Assume the algorithm reached Line \ref{line:lastOutputu}, and let $x'$ be the partial assignment $x^*$ constructed by the algorithm when it reaches Line \ref{line:lastOutputu}.

  Suppose, for the sake of contradiction, there exists an input $y\in \uf^{-1}(1)$ that is consistent with $x'$. Since the algorithm reached Line \ref{line:lastOutputu}, it must be the case that the partial assignment $x'$ constructed does not contain a $b$-certificate for any $b\in\zuo$ because otherwise one of the \textbf{if} conditions between Lines \ref{line:Ifu} and \ref{line:If0} would have terminated the algorithm. Then using Observation \ref{obs:existsucert}, there must also exist a $\un$-input consistent with $x'$. Hence every time Line \ref{line: firstCertificate1} was executed, there was indeed an $x \in \uf^{-1}(\un)$ consistent with $x^*$. Suppose the $\un$-certificates used during the run of the \textbf{for} loop on Line \ref{line:For} were $c_1, \ldots, c_k$, and their respective domains were $C_1,\ldots, C_k$. 

 The fact that neither of the \textbf{if} conditions fired means that every time the oracle was queried, the replies differed from the $\un$-certificate being queried in at least one index each time. Let $B_i\subseteq C_i$ be the set of positions in $C_i$ where $x'$ differs from $c_i$. By the observations above, each $B_i$ is non-empty. Observe that since $c_{i+1}$ is chosen to be consistent with $x^*$, it must be the case that $c_{i+1}$ and $x^*$ agree on all positions in $C_i$. Hence $B_{i+1}$ is disjoint from $B_{i}$. With the same reasoning, we can conclude that $B_{i+1}$ is disjoint from every $B_j$ where $j\le i$.

Observe that for each $i\in [k]$, there is a setting to the bits in $B_i$ such that $x'$ becomes a $\un$-input -- simply take the setting of these bits from $c_i$, which is a $\un$-certificate. More formally, for each $i\in [k]$, there exists a string $\alpha_i\in \zo^{|B_i|}$ such that $f(x'|_{B_i\gets \alpha_i})=\un$. 

Since $y$ is consistent with $x'$, it agrees with $x'$ on all positions where $x'\neq *$. This means the previous observation holds for $y$ too. That is, for each $i\in [k]$, there exists strings $\alpha_i\in \zo^{|B_i|}$ such that $f(y|_{B_i\gets \alpha_i})=\un$. Recall that $y\in \uf^{-1}(1)$, and hence the sets $B_1,\ldots, B_k$ form a collection of disjoint sensitive blocks for $\uf$ at $y$. Further, since the algorithm has not found a $1$-certificate (or $0$-certificate) yet, it must be the case that there is some $\un$-input consistent with $x'$ by Observation~\ref{obs:existsucert}. This means there is yet another disjoint block $B_{k+1}$ that is sensitive for $\uf$ at $y$. But this is a contradiction since the maximum, over all $1$-inputs of $\uf$, number of disjoint sensitive blocks is $\bsen_{\un,1}(\uf) = k < k+1$.

A nearly identical proof can be used to show that every $0$-input is inconsistent with $x'$.
\end{proof}

\begin{theorem}\label{thm: algo upper bound}
    Algorithm~\ref{algo:UpperBoundu} correctly computes $\uf$, and makes at most $O(\cc_\un(f)\bsen_\un(f))$ queries. Thus $\depth_\un(f) = O(\cc_\un(f) \cdot \bsen_\un(f))$.
\end{theorem}
\begin{proof}
    If the algorithm outputs a value in $\zo$, then it must have passed the corresponding \textbf{if} condition (either in Line~\ref{line:If0}, or in Line~\ref{line:If1} and is trivially correct. If the algorithm outputs $\un$, then either the \textbf{if} condition on Line \ref{line:Ifu} must have passed, or Line \ref{line:lastOutputu} must have been reached. In the former case, the correctness of the algorithm is trivial. In the latter case, from Claim \ref{lem:firstInconsistency1}, we conclude that every $0$-input and every $1$-input of $\uf$ must be inconsistent with the partial assignment $x^*$ constructed by the algorithm when it arrives at Line \ref{line:lastOutputu}. This means that every $x \in \zuo^n$ that is consistent with $x^*$ (in particular, the unknown input $x$) must satisfy $\uf(x) = \un$, which concludes the proof of correctness.

    The \textbf{for} loop runs for $\max(\bsen_{\un,0}(f), \bsen_{\un,1}(f))= O(\bsen_\un(f))$ many iterations, and at most $\cc_\un(f)$ many bits are queried in each iteration.
\end{proof}

\subsection{Improved Bounds for Sensitivity and Certificate Complexity}
\label{app:characterize-sen}

In this subsection, we explore the $\un$ complexity measures further and initially prove the following characterization for $\un$-sensitivity. We use $\sen_\un^{(b)}$ where $b\in\zuo$ to denote the $\un$-sensitivity for inputs that yield a $b$ output.
\begin{theorem}\label{thm:senstr}
    Let $\uf$ be the hazard-free extension of some Boolean function $f$. $\sen_\un^{(1)}(\uf, x)$ is maximized when $x$ is a prime implicant, $\sen_\un^{(0)}(\uf, x)$ is maximized when $x$ is a prime implicate, and $\sen_\un^{(\un)}(\uf, x)$ is maximized when $x$ contains exactly one $\un$.
\end{theorem}
\begin{proof}
    Consider $x$ which is an implicant but not a prime implicant. Then, there is a prime implicant $y < x$. We have already seen that the stable bits in $y$ are sensitive. So we just have to argue that if bit $i$ is stable in $x$ and unstable in $y$, then $i$ is not sensitive for $x$. Suppose $x_i=0$ and $y_i=\un$. Changing $x_i=\un$ cannot change the output as $y$ has the same output as $x$. Changing $x_i=1$ cannot change the output as then $y$ would not be a prime implicant.

    A similar argument as above holds for the prime implicates.

    For the last part, consider an $x$ that has more than one $\un$ such that $f(x)=\un$. Then, there is a $y > x$ that has exactly one $\un$ such that $\uf(y)=\un$ ($y$ is an edge in the hypercube). Let the strings $w$ and $v$ in $\zon$ be the two proper resolutions of $y$ and $j$ be the unique position where $w$ and $v$ differ. We argue that the sensitive bits for $x$ are also sensitive bits for $y$. Bit $j$ is sensitive for $y$. So consider an $i\neq j$ such that bit $i$ is sensitive for $x$. Note that $y_i$ is a stable value. If $x_i=\un$, then it must be that setting $x_i=y_i$ or $x_i=\neg{y_i}$ changes the output to a stable value. Setting $x_i=y_i$ cannot yield a stable output as the resulting subcube will still contain the edge $y$. So, setting $x_i=\neg{y_i}$ yields a stable output. This implies that changing the bit $i$ of $y$ to $\neg{y_i}$ will yield a stable output, proving that bit $i$ is stable for $y$ as well. Now, we consider the cases where $x_i$ is stable. It must be that $x_i=y_i$ ($x_i$ cannot be $\neg{y_i}$ as $y>x$) and say $i$ is sensitive for $x$. Since $f$ is natural, it must be that setting $x_i$ to $\neg{y_i}$ yields a stable output. But then, changing bit $i$ in $y$ to $\neg{y_i}$ will also yield a stable output proving that bit $i$ is sensitive for $y$ as well.
\end{proof}

The following lemma upper bounds the hazard-free sensitivity for $\un$-inputs using Boolean sensitivity, denoted $\sen(f)$ for $f$.

\begin{lemma}\label{lem:usen-sen}
    For any Boolean function $f$, we have $\sen_\un^{(\un)}(f) \leq 2\sen(f)-1$.
\end{lemma}
\begin{proof}
    From Theorem ~\ref{thm:senstr}, we know that $\sen_\un^{(\un)}$ is maximized at an input $x\in\zuon$ that has exactly one $\un$. Let $v$ and $w$ be the two proper resolutions of $x$ such that $f(v) = 0$ and $f(w) = 1$ and they differ only in bit $i$. We will prove that all sensitive bits of $x$ for $\uf$ are sensitive for either $v$ or $w$ for $f$. Clearly, bit $i$ is sensitive for $v$ and $w$ for $f$. Consider a $j\neq i$ that is sensitive for $x$ for $\uf$. Notice $x_j$ must be stable and $x_j=v_j=w_j$. Say, $x_j=0$. Since the output on $x$ of $\uf$ is already $\un$, the reason bit $j$ must be sensitive is that setting $x_j=1$ causes $\uf$ to output a stable value, say $0$. This implies that if we set $w_j=1$ in $w$, then the output of $f$ changes from $1$ to $0$. So, bit $j$ must be stable for $w$ for $f$. The rest of the cases are similar. The sensitive bits of $x$ for $\uf$ is contained in the union of sensitive bits of $v$ and $w$ for $f$. The union has at least $i$ as the common bit. The lemma follows.
\end{proof}

Similar to the Theorem \ref{thm:senstr} we show the following characterization for $\un-$ certificate complexity. We use $\cc_\un^{(b)}$ where $b\in\zuo$ to denote the $\un$-certificate complexity for inputs that yield a $b$ output.

\begin{theorem}\label{thm:ccstr}
    Let $\uf$ be the hazard-free extension of some Boolean function $f$. $\cc_\un^{(1)}(\uf, x)$ is maximized when $x$ is a prime implicant, $\cc_\un^{(0)}(\uf, x)$ is maximized when $x$ is a prime implicate and $\cc_\un^{(1)}(f, x) = k_1, \cc_\un^{(0)}(\uf, x) = k_2$ where $k_1, k_2$ are the sizes of the largest prime implicant and implicate respectively.
\end{theorem}

\begin{proof}
    Observation~\ref{obs: unstable certificate no u} implies that for any input $x$ giving a stable output $b$, a minimum certificate does not have any unstable bits in it. removing any of the stable bits in the minimum certificate also no longer makes it a certificate. Hence the literals that form the certificate give a prime implicant (or implicate) if $\uf(x) = 1$ (or 0) respectively. This implies that $\cc_\un^{(0)}(\uf) \le k_2$ and $\cc_\un^{(1)}(\uf) \le k_1$.
    
    Now consider the largest prime implicant $P$ of $f$. The input $x_P$ (defined to be $1$ in those positions where $P$ contains the corresponding positive literal, $0$ where $P$ has the corresponding negated literal, and $\un$ everywhere else) has it's minimum certificate to be all of it's stable bits. This happens since we know that for such an input all the stable bits are sensitive and any certificate should contain all the stable bits of an input. Hence $\cc_\un^{(1)}(\uf, x) = k_1 = \sen_\un^{(0)}(\uf, x), \cc_\un^{(0)}(\uf, x) = k_2 = \sen_\un^{(0)}(\uf, x)$.

    An upper bound of $k_1 + k_2 -1$ for $\cc_\un^{(0)}(\uf, x)$ has been proved as part of the Sensitivity theorem.
\end{proof}

We end this subsection by showing that the bounds in Lemma~\ref{lem:usen-sen} and part~2 of Theorem ~\ref{thm:ccstr} are tight for the majority function.

The $\fmaj$ function on $2n+1$ variables has:
\begin{description}
    \item [$\sen_\un^{(u)}(\fmaj_n) = 2n+1$] Consider the input $x$ with $n$ 0's followed by a $u$, followed by $n$ 1's. Note that flipping any of the 0's to 1 makes the function 1 and flipping any of the 1's to a 0 makes the function 1. Also, changing the $\un$ to a $0$ or a $1$ makes the output $0$ or $1$.
    \item[$\sen_\un^{(b)}(\fmaj_n) = n+1$ for $b \in \zo$] An input $x$ such that $\widehat{\fmaj}(x) = b$ has at least $n+1$ b's appearing in $x$. If there are more than $n+1$ b's none of the bits are sensitive. If there are exactly $n+1$ b's then exactly these bits are sensitive.
    \item [$\sen(\fmaj_n) = n+1$] Consider the input with $n+1$ ones followed by $n$ zeroes. All positions with a one are sensitive and all positions with a zero are not sensitive. If an input has more than $n+1$ zeroes or ones, then none of the bits are sensitive.
\end{description}

\section{Deterministic, Quantum and Randomized Models}


In this section, we first present the relationship between the deterministic, randomized and  quantum $\un$-query complexities. We derive exact bounds for the $\fmux_n$ function in the three variants of the model. Further, we also present a relationship between the query complexity with and without uncertainty for monotone functions. 

\subsection{Relationships between Query Complexities}

In this section, we show that deterministic, randomized, and quantum $\un$-query complexities are all polynomially related. We use Yao's minimax principle~\cite{Yao77} and the adversary method for lower bounds on quantum query complexity due to Ambainis~\cite{Amb02}. We state them below for convenience.
\begin{lemma}[Yao's minimax principle]\label{lem: yao}
    For finite sets $D, E$ and a function $f: D \to E$, we have $\rand(f) \geq k$ if and only if there exists a distribution $\mu : D \to [0,1]$ such that $\depth_\mu(f) \geq k$.
    Here, $\depth_\mu(f)$ is the minimum depth of a deterministic decision tree that computes $f$ to error at most $1/3$ when inputs are drawn from the distribution $\mu$.
\end{lemma}

\begin{theorem}[{\cite[Theorem~5.1]{Amb02}}]\label{thm: ambainis}
    Let $D, E$ be finite sets, let $n$ be a positive integer, and let $f : D^n \to E$ be a function. Let $X, Y$ be two sets of inputs such that $f(x) \neq f(y)$ for all $(x, y) \in X \times Y$. Let $R \subseteq X \times Y$ be such that
    \begin{itemize}
        \item For every $x \in X$, there are at least $m$ different $y \in Y$ with $(x, y) \in R$,
        \item for every $y \in Y$, there are at least $m'$ different $x \in X$ with $(x, y) \in R$,
        \item for every $x \in X$ and $i \in [n]$, there are at most $\ell$ different $y \in Y$ such that $(x, y) \in R$ and $x_i \neq y_i$,
        \item for every $y \in Y$ and $i \in [n]$, there are at most $\ell'$ different $x \in X$ such that $(x, y) \in R$ and $x_i \neq y_i$.
    \end{itemize}
    Then $\qua(f) = \Omega(\sqrt{\frac{mm'}{\ell\ell'}})$.
\end{theorem}

We lower bound randomized and quantum $\un$-query complexities polynomially by $\un$-block sensitivity.
\begin{lemma}\label{lem: unstable randomized quantum bs}
    Let $f : \zon \to \zo$. Then,
    \[
    \rand_\un(f) = \Omega(\bsen_\un(f)), \qquad \qua_\un(f) = \Omega(\sqrt{\bsen_\un(f)}).
    \]
\end{lemma}
\begin{proof}
    We use Lemma~\ref{lem: yao} for the randomized lower bound. Let $x \in \zuo^n$ be such that $\bsen_\un(f) = \bsen_\un(\uf, x) = k$, with corresponding sensitive blocks $B_1, \dots, B_k$. Define a distribution $\mu$ on $\zuo^n$ as follows:
    \begin{align*}
        \mu(x) & = 1/2,\\
        \mu(x^{B_i}) & = 1/2k~\textnormal{for all}~i \in [k].
    \end{align*}
    Towards a contradiction, let $T$ be a deterministic decision tree of cost less than $k/10$ that computes $\uf$ to error at most $1/3$ under the input distribution $\mu$. Let $L_x$ denote the leaf of $T$ reached by the input $x$. There are now two cases:
    \begin{itemize}
        \item If the output at $L_x$ is not equal to $\uf(x)$, then $T$ errs on $x$, which contributes to an error of 1/2 under $\mu$, which is a contradiction.
        \item If the output at $L_x$ equals $\uf(x)$, since the number of queries on this path is less than $k/10$, there must exist at least $9k/10$ blocks $B_i$ from the set $B_1, \dots, B_k$ such that no variable of $x^{B_i}$ is read on input $x^{B_i}$ (observe that in this case, $x^{B_i}$ also reaches the leaf $L_x$). Since $f(x^{B_i}) \neq f(x)$, this means $T$ makes an error on each of these $B_i$'s, contributing to a total error of at least $9k/10 \cdot 1/2k = 0.45$ under $\mu$, which is a contradiction.
    \end{itemize}
    This concludes the randomized lower bound.

    For the quantum lower bound we use Theorem~\ref{thm: ambainis}. Define $X = \cbra{x}, Y = \cbra{x^{B_i} : i \in [k]}$, and $R = X \times Y$. From Theorem~\ref{thm: ambainis} we have $m = k, m' = 1, \ell = 1$ (since each index appears in at most 1 block, as each block is disjoint) and $\ell' = 1$. Theorem~\ref{thm: ambainis} then implies $\qua_\un(f) = \Omega(\sqrt{k}) = \Omega(\sqrt{\bsen_\un(f)})$.
\end{proof}

\begin{theorem}
  \label{thm:hazard-free-relationships}
    For $f : \zon \to \zo$, we have
    \[
    \depth_\un(f) = O(\rand_\un(f)^2), \qquad \depth_\un(f) = O(\qua_\un(f)^4).
    \]
\end{theorem}

\begin{proof}
    Theorem~\ref{thm: algo upper bound} and Theorem~\ref{thm:sbscc} imply
    \[
        \depth_\un(f) = O(\cc_\un(f) \cdot \bsen_\un(f)) = O(\bsen_\un(f)^2) = O(\rand_\un(f)^2),
    \]
    where the final bound is from Lemma~\ref{lem: unstable randomized quantum bs}. Substituting the second bound from Lemma~\ref{lem: unstable randomized quantum bs} in the last equality above yields $\depth_\un(f) = O(\qua_\un(f)^4)$.
\end{proof}

\subsection{Exact Bounds for $\fmux_n$}

It is easy to prove using Theorem~\ref{thm:sbscc} that the $\un$-query complexity of $\fmux_n$ is exponentially larger than its query complexity. We claim that $\sen_\un(\fmux_n) \geq 2^n$. Consider the input where all selector bits are $\un$ and all data bits are $1$. Flipping any data bit to zero changes the output from $1$ to $\un$.  In the following theorem, we prove a stronger statement. A function is said to be \emph{evasive} if its query complexity is the maximum possible. We first prove Theorem~\ref{thm:muxdepth}, i.e., $\fmux_n$ is evasive for all $n$. In fact, we prove a more general theorem. We consider decision trees for $\widetilde{\fmux_n}$ that are guaranteed to produce the correct output when the number of unstable values in the input is at most $k$, for an arbitrary $k\in[0,2^n+n]$. 

\begin{theorem}\label{thm:muxevasive}
Let $k \in [0, 2^n+n]$. Any optimal-depth decision tree that correctly computes the hazard-free extension of $\fmux_n$ on inputs with at most $k$ unstable values has depth exactly $\min{\{2^k + n, 2^n + n\}}$.
\end{theorem}
\begin{proof}
    We first prove the theorem for $k \in [0, n]$.
    
    We prove the lower bound using the adversarial strategy in Algorithm~\ref{algo:advmux}.
     \begin{algorithm}[ht]
        \begin{algorithmic}[1]
            \State{\textbf{Given: } $k \in [0, n]$}
            \State $C_s \gets \emptyset, Q_d \gets \emptyset, Q_s \gets \emptyset$
            \State $U_s \gets \{i \mid i \in [n]\}$
            \State $U_d \gets \{x_{(b_0, \dotsc, b_{n-1})} \mid b_i \in \zo \}$
            \State $r \gets \underbrace{(1, \dotsc, 1)}_{n\text{ 1s}}$
            \Loop
            \If{query is for a selector bit $s_i$}
                \State{$Q_s \gets Q_s \cup \{i\}$}
                \If{$|Q_s| \leq k$} \label{line:limk}
                    \State \Return $\un$
                \Else
                    \State{$C_s \gets C_s \cup \{ i \}$}
                    \State{$U_s \gets U_s \setminus \{i\}$}
                    \State{Remove all data bits $x_{(b_0,\dotsc,b_{n-1})}$ from $U_d$ where $b_i \neq r_i$ \label{line:ud}}
                    \State \Return $r_i$
                \EndIf
            \Else
                \State{Let $x_{(b_0,\dotsc,b_{n-1})}$ be the queried data bit.} \label{line:d1}
                \If{$U_d\setminus Q_d = \{x_{(b_0,\dotsc,b_{n-1})}\}$} \label{line:d1if}
                    \State{$r \gets (b_0,\dotsc,b_{n-1})$} \label{line:rchange}
                    \State{$Q_d \gets Q_d \cup \{x_{(b_0,\dotsc,b_{n-1})}\}$}
                    \State \Return{$0$}
                \Else
                    \State{$Q_d \gets Q_d \cup \{x_{(b_0,\dotsc,b_{n-1})}\}$}
                    \State \Return{$1$}
                \EndIf \label{line:d2}
            \EndIf
            \EndLoop
            \caption{Adversarial $\un$-query responder for $\fmux_n$}
            \label{algo:advmux}
        \end{algorithmic}
    \end{algorithm}
    
    At any point of time during the play of the game, we associate a partial assignment with the $n$ selector bits where a $\ast$ at position $i$ indicates that the selector bit $s_i$ has not been queried yet. This partial assignment has at most $k$ positions of value $\un$ by line~\ref{line:limk}. We observe the following invariants about the variables in the algorithm.
    \begin{enumerate}
        \item $Q_d$ is the set of all queried data bits.
        \item $Q_s$ is the set of all queried selector bits.
        \item $C_s$ is the set of all queried selector bits for which we returned $0$ or $1$.
        \item $U_s$ is the set of all selector bits which were either not queried or queried and for which we returned $\un$.
        \item $U_d$ is the set of all data bits $x_{(b_0,\dotsc,b_{n-1})}$ such that the there is an assignment $x \in \zuon$ to the selector bits that is consistent with the current partial assignment and $(b_0, \dotsc, b_{n-1}) \in \res(x)$.
        \item The variable $r$ and the set $U_d$ are such that $x_{r} \in U_{d}$ is always true.
    \end{enumerate}
    
    The first 4 observations are straightforward. For 5, notice that initially $U_d$ is the set of all data bits and line~\ref{line:ud} ensures that we remove the data bits that have indices inconsistent with the current partial assignment whenever we return a $0$ or $1$ for a selector bit. Intially, $r = (1,\dotsc,1)$ and $x_{(1,\dotsc,1)} \in U_d$. We remove $x_{(1,\dotsc,1)}$ from $U_d$ only if some $r_i = 0$. This is only possible once we have changed the value of $r$. This only occurs in line~\ref{line:rchange}. This line sets the value of $r$ to the index of some data bit in $U_d$. So this property still holds. After this, as we will argue below, we never change the value of $r$, and whenever we return a value of $0$ or $1$ for a selector bit, we make sure it is consistent with $r$. So the data bit indexed by $r$ never gets removed from $U_d$.
    
    We now prove the theorem using a series of claims.
    \begin{claim}\label{claim:datainout}
        Suppose $Q_s \neq [n]$ and $x_{(b_0, \dotsc, b_{n-1})}\in U_d$. Then:
        \begin{enumerate}
            \item There is an assignment $x \in \zuon$ to the selector bits with at most $k$ positions with value $\un$ such that $x$ is consistent with the current partial assignment and $(b_0,\dotsc,b_{n-1}) \not\in \res(x)$.
            \item There is an assignment $y \in \zuon$ to the selector bits with at most $k$ positions with value $\un$ such that $y$ is consistent with the current partial assignment and $(b_0,\dotsc, b_{n-1}) \in \res(y)$.
        \end{enumerate}
    \end{claim}
    \begin{proof}
    \begin{enumerate}
        \item Suppose $s_j$ is an unqueried selector bit. Consider an assignment $x\in\zuon$ that is consistent with the current partial assignment constructed as follows: we set $x_j = \neg b_j$ for all unqueried selector bits. This assignment has at most $k$ positions with value $\un$ as we extend the partial assignment only with values in $\zo$ and $(b_0,\dotsc,b_{n-1}) \not\in \res(x)$ as there is at least one unqueried selector bit.
        \item Construct an assignment $y\in\zuon$ for the selector bits by extending the partial assignment as follows: for any unqueried selector bit $s_j$, we set $y_j = b_j$. Again, $y$ has at most $k$ positions with value $\un$ as we extend only using values in $\zo$. We also have $(b_0,\dotsc,b_{n-1}) \in \res(x)$ by the property of data bits in $U_d$ stated earlier.
    \end{enumerate}
    \end{proof}
    
    \begin{claim}\label{claim:nounqueriedud}
        When the querier outputs an answer, the set $U_d \setminus Q_d$ must be empty.
    \end{claim}
    \begin{proof}
        Suppose for contradiction that there is an unqueried data bit $x_{(b_0,\cdots b_{n-1})}$ in $U_d$. Observe in lines~\ref{line:d1} to lines~\ref{line:d2} that this implies we replied $1$ for all queries to data bits. We construct an assignment $x\in\zuon$ of the selector bits that is consistent with the partial assignment as follows: for any unqueried $s_i$, set $x_i = b_i$. Since we only use $0$ or $1$ to extend the partial assignment, $x$ has at most $k$ positions with value $\un$. Also, $(b_0,\cdots b_{n-1}) \in \res(x)$ by construction.
    
        If the querier says that the answer is $1$, we can set $x_{(b_0,\cdots b_{n-1})}$ to $0$, the selector bits to $x$ and make the correct answer $0$ or $\un$. If the querier says $0$ or $\un$, we can set all data bits to $1$, the selector bits to $x$ and make the correct answer $1$.
    \end{proof}

    \begin{claim}\label{claim:allsel}
        When the querier outputs an answer, the set $Q_s$ must be $[n]$.
    \end{claim}
    \begin{proof}
        Suppose for contradiction that some $s_i$ is unqueried. By Claim~\ref{claim:nounqueriedud}, we know that there are no unqueried data bits in $U_d$. Say $x_{(b_0,\cdots b_{n-1})}$ was the last queried data bit in $U_d$. By claim~\ref{claim:datainout} part~2, we know $(b_0, \dotsc, b_{n-1}) \in \res(y)$ for some $y\in\zuon$ such that $y$ is consistent with the current partial assignment and has at most $k$ positions with value $\un$. We know from lines~\ref{line:d1} to \ref{line:d2} that $x_{(b_0,\cdots b_{n-1})}$ was set to $0$ when it was queried. Observe that we set the value of $r$ to $(b_0,\cdots b_{n-1})$ when this happens and we reply with $r_i$ for the selector bit $s_i$ whenever we reply a $0$ or a $1$ for the selector bits. Also, the value of $r$ can never change after this as there cannot be another unqueried bit in $U_d$ as the algorithm only removes elements from $U_d$. Therefore, the data bit $x_{(b_0,\dotsc,b_{n-1})}$ can never be removed from $U_d$ after it has been set to $0$. Also, every other data bit is either unset or set to $1$.
        
        If the querier answers $1$, we set the selector bits to $y$, the remaining data bits to $1$ and then the correct answer is $0$ or $\un$. If the querier answers $0$ or $\un$, we use part~1 of Claim~\ref{claim:datainout} and use that $x$ as values of selector bits, set all the remaining data bits to $1$. The data bit $x_{(b_0,\dotsc,b_{n-1})} \not\in \res(x)$. So the correct answer is $1$.
    \end{proof}
    
    Since all the selector bits have to be queried by Claim \ref{claim:allsel} and exactly $k$ of them have been set to $\un$, we know that $|U_d| = 2^k$. By Claim \ref{claim:nounqueriedud} we know that all of the data bits in $U_d$ have to be queried and hence the querier must have queried at least $2^k + n$ bits.
        
    For the upper bound, first the querier queries all the $n$ selector bits. The reply can have a maximum of $k$ positions with value $\un$. This means the function only depends on at most $2^k$ data bits. Query all those data bits. If they are all $0$ or $1$, then output $0$ or $1$. Otherwise, output $\un$.
    
    If $k > n$, then the claimed lower bound is $2^n+n$ which holds as it is a lower bound for even inputs with at most $n$ positions with value $\un$. The claimed upper bound is also $2^n+n$ which holds as we can compute any function if we are allowed to query all the input bits.
\end{proof}

\begin{lemma}
\label{lemm:muxRandQuant}
\[ \qquad \rand_\un(\fmux_n) = \Theta(2^n), \qquad \qua_\un(\fmux_n) = \Theta(2^{n/2}). \]
\end{lemma}
\begin{proof}
Let $f = \fmux_n$. We first show the upper bounds. The randomized upper bound is trivial since the number of input bits is $\Theta(2^n)$. For the quantum upper bound, we first query the $n$ selector bits, let the output be $x \in \zuo^n$. Further, let $z \in \zuo^{2^n}$ denote the input's data bits. Consider the set of indices $I = \cbra{y = (y_1, \dots y_n) : y \in \res(x)}$. By the definition of $\fmux_n$, we have $\uf(x, z) = b \in \zo$ iff $z_i = b$ for all $i \in I$. By a standard use of Grover's search algorithm a constant number of times, we can determine with high probability whether $z$ is the constant 0 on $I$, or whether it is the constant 1 on $I$, or whether neither of these is the case. This immediately gives us the output: 0 in the first case, 1 in the second case, and $\un$ in the last case. The cost of querying the selector bits is $n$, and constantly many applications of Grover's search algorithm costs $O(2^{n/2})$.
    
    We now show the lower bounds using a reduction from the OR function over $2^n$ bits that we denote $\OR_{2^n}$. This function is well known to admit a randomized lower bound of $\Omega(2^n)$ and a quantum lower bound of $\Omega(2^{n/2})$, even in the absence of uncertainty. Given a (randomized/quantum) algorithm $\mathcal{A}$ solving $\uf$, we construct a (randomized/quantum) algorithm $\mathcal{A}'$ with query complexity at most that of $\mathcal{A}$, and that solves $\OR_{2^n}$. Let the input to $\OR_{2^n}$ be denoted by $x$.

    We define $\A'$ to be the algorithm that simulates a run of $\mathcal{A}$ in the following way:
    \begin{enumerate}
        \setlength\itemsep{-.3em}
        \item If $\A$ queries the $j$'th bit of its input where $j \in [n]$, return $\un$.
        \item If $\A$ queries the $j$'th bit of its input, where $j = n + k$ for some $k > 0$, then query $x_k$ and return that value.
    \end{enumerate}
    If $\A$ outputs $0$, $\A'$ outputs $0$. Else $\A'$ outputs $1$. This can be simulated in the randomized and quantum models.
    Note that the query complexity of $\A'$ is at most that of $\A$. We now argue correctness.
    
    \begin{itemize}
        \item Let $x = 0^{2^n}$. The input $z \in \zuo^{n + 2^n}$ to $\uf$ with all addressing bits set to $\un$, and the data bits all equal to $0$ is consistent with the run of $\A$ above. Note that for all $y \in \res(z)$ we have $f(y) = 0$ since the relevant data bit is always 0 regardless of the choice of selector bits. 
        Since (with high probability) $\A$ must be correct on $z$, this means that the output of $\A$, and hence $\A'$, is 0 (with high probability) and correct in this case.
        \item Let $x \neq 0^{2^n}$. The input $z \in \zuo^{n + 2^n}$ to $\uf$ with all addressing bits set to $\un$, and the data bits equal to the string $x$ is consistent with the run of $\A$ above. Now there exists a $y \in \res(z)$ that sets all selector bits to point to a non-0 bit of the data bits (since there exists such a data bit). By the correctness of $\A$ (with high probability), this means the output of $\A$ must either be $\un$ or 1. Thus $\A'$ outputs 1 in this case (with high probability), which is the correct answer.
    \end{itemize}
\end{proof}


\subsection{Relating Query Complexity with and without Uncertainty for Monotone Functions}

In this section, we prove that for monotone functions, the presence of uncertainty does not make computation significantly harder. Similar relationships are known for monotone combinational circuits w.r.t.~containing metastability \cite{ikenmeyer2019complexity}.
\begin{proof}[Proof of Lemma~\ref{lem:monsame}]
    Any query algorithm for $\uf$ also computes $f$ with at most the same cost, and hence $\depth(f) \le \depth_\un(f)$. Similarly we have $\rand(f) \leq \rand_\un(f)$ and $\qua(f) \leq \qua_\un(f)$.

    We start with a best (deterministic/randomized/quantum) query algorithm $A$ for $f$, and an oracle holding an input $x \in \zuon$ for $\uf$ .
    Now, for $b \in \zo$, we define $A_b$ to be the same algorithm as $A$, but whenever $A$ queries the $j^\text{th}$ bit of its input, it performs the following operation instead:
    
    \begin{enumerate}
        \item Query the $j$'th bit of $x$, denote the outcome by $x_j$.
        \item If $x_j \in \{0,1\}$, return $x_j$.
        \item If $x_j = \un$, return $b$.
    \end{enumerate}

    In the case of quantum query complexity, note that this operation can indeed be implemented quantumly making $2$ queries to $O_x$. The initial query performs the instructions described above, and the second query uncomputes the values from the tuple we don't need for the remaining part of the computation. Note here that $O_x^3 = I$ by definition, and thus $O_x^2 = O_x^{-1}$, which is what we need to implement for the uncomputation operations.
    
    Let $S_{\un} = \{i\mid x_i=\un \}$ be the positions that have $\un$ in $x$. Recall that $y^0 := x|_{S_u\leftarrow \vec{0}}$ (and $y^1 := x|_{S_u\leftarrow \vec{1}}$) is the input that has all the $\un$ in $x$ replaced by $0$s (by $1$s respectively). Run $A_0$ and $A_1$ (possibly repeated constantly many times each) to determine the values of $f(y^0)$ and $f(y^1)$ with high probability. If $f(y^0)=1$, then we output $1$. Else if $f(y^1)=0$, we output $0$. Else we have $f(y^0)=0$ and $f(y^1)=1$, and we output $\un$.
    
    \textbf{Correctness:} First observe that for $b \in \zo$, all answers to queries in the algorithm $A_b$ are consistent with the input $y^b$. Next observe that in the poset (equivalently `subcube') formed by the resolutions of $x$, the inputs $y^0$ and $y^1$ form the bottom and top elements respectively. Since $f$ is monotone, if $f(y^0)=1$, we can conclude that $f$ is $1$ on all resolutions of $x$. Similarly when $f(y^1)=0$, it must be the case that $f$ is $0$ on all inputs in the poset. The remaining case is when $f(y^0)=0$ and $f(y^1)=1$. In this case, the inputs $y^0$ and $y^1$ are themselves resolutions of $x$ with different evaluations of $f$, and hence the algorithm correctly outputs $\un$. 

    By standard boosting of success probability using a Chernoff bound by repeating $A_0$ ($A_1$) constantly many times and taking a majority vote of the answers, we can ensure that the correctness probability of $A_0$ ($A_1$) is large enough, say at least 0.9. Thus, the algorithm described above has correctness probability at least $0.9^2 = 0.81$, and its cost is at most a constant times the cost of $A$.
\end{proof}

\subsection{A Non-Monotone Easy Function}\label{app:nonmoneasy}

In the introduction, we exhibited a monotone function on $n + \binom{n}{n/2}$ bits, with a query complexity of $n+1$. We now show that this is not a consquence of monotonicity by describing a non-monotone, non-degenerate $n$-variate function such that its hazard-free extension has $O(\log n)$ query complexity.

\begin{theorem}
    For any $n \geq 1$, there is a non-monotone, non-degenerate Boolean function on $n + 2^{n+1} - 1$ inputs such that the query complexity of its hazard-free extension is $2n+1$.
\end{theorem}
\begin{proof}
    The function has $n$ input variables $\mathbf{s} = (s_0, \dotsc, s_{n-1})$ and $2^{n+1}-1$ input variables $y = (y_{(b_1, \dotsc, b_i)})$ where $i\in [0, n+1]$ and each $b_j\in\{0,1\}$ (The index set for $y_{*}$ variables is sequences of bits of length $0$ through $n+1$):
    \begin{equation*}
        f_n(\mathbf{s}, \mathbf{y}) = \fmux_n(s_0, \dotsc, s_{n-1}, g_{(0, \dotsc, 0)}(\mathbf{y}), \dotsc, g_{(1, \dotsc, 1)}(\mathbf{y}))
    \end{equation*}
    where we define:
    \begin{equation*}
        g_{(v_0, \dotsc, v_{n-1})} = z_{\mathbf{v}0}\ \mathsf{op}_{\mathbf{v}0}\  (z_{\mathbf{v}1}\ \mathsf{op}_{\mathbf{v}1}\ (z_{\mathbf{v}2} \dotsc \mathsf{op}_{{\mathbf{v},n-1}}\ z_{\mathbf{v}n})\dotsc)
    \end{equation*}
    where $\mathbf{v} = (v_0, \dotsc, v_{n-1})$, $\mathsf{op}_{\mathbf{v}j}$ is Boolean $+$ if $v_j = 0$ and Boolean $\cdot$ otherwise. Each $z_{\mathbf{v}j}$ is the variable $y_{(v_0,\dotsc,v_{j-1})}$. In particular, $z_{\mathbf{v}0}$ is the variable $y_{()}$. i.e., it is independent of $\mathbf{v}$.

    We prove the theorem by induction on $n$. For $n=1$, the function $f_1$ is (after rewriting variables for convenience): $f_1(s, x, y, z) = \fmux_1(s, x+y, xz)$. A hazard-free decision tree of depth $3$ for $f_1$ is as follows: First, we query $s$. If it is $0$ or $1$, then two more queries suffice. Suppose $s = \un$. Then, we query $x$. If $x=0$, then $xz=0$ and the answer can be determined by only querying $y$. If $x=1$, then $x+y=1$ and the answer can be determined by only querying $z$. If $x=\un$, then we can immediately answer $\un$ since $x+y$ can only be $\un$ or $1$ and $xz$ can only be $\un$ or $0$.

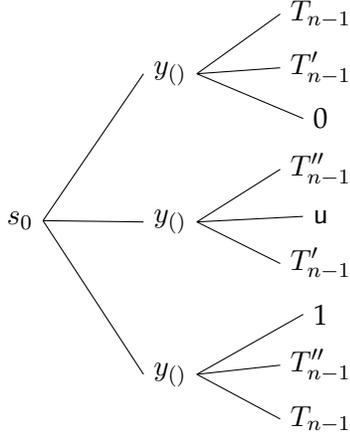
\begin{figure}
\caption{Hazard-free decision tree for $f_n$\label{fig:smalldepth}}
\begin{tikzpicture}[grow=right]
    \tikzset{level distance=2cm}
    \Tree [.{$s_0$}
      [.{$y_{()}$} [.{$T_{n-1}$} ] [.{$T''_{n-1}$} ] [.1 ] ]
      [.{$y_{()}$} [.{$T'_{n-1}$} ] [.{$\un$} ] [.{$T''_{n-1}$} ]]
      [.{$y_{()}$} [.0 ] [.{$T'_{n-1}$} ] [.{$T_{n-1}$} ] ]
    ]
\end{tikzpicture}
\end{figure}

    The proof for the inductive case is similar to the base case. For $f_n$, in two queries, we either determine the answer or reduce the remaining decision to the $n-1$ case. The decision tree is given in Figure~\ref{fig:smalldepth}. We split the proof of correctness of our construction into nine different cases corresponding to the leaves of this tree.
    
    ($s_0=0$ and $y_{()}=0$) In this case, we know that the selector bits can only select functions $g_{\mathbf{v}}$ such that $\mathsf{op}_{\mathbf{v}0} = +$. Since, $y_{()}=0$, the remaining problem is isomorphic to $f_{n-1}$ on the selector bits $(s_1, \dotsc, s_{n-1})$ and the functions $g_{\mathbf{v}}$ where $v_0=0$ and the variable $y_{()}$ has been substituted with $0$.

    ($s_0=0$ and $y_{()}=\un$) We know that the selector bits can only select functions $g_{\mathbf{v}}$ such that $v_0=0$ which implies $\mathsf{op}_{\mathbf{v}0} = +$. Since $y_{()}=\un$, the function $f_n$ can only now evaluate to $\un$ or $1$. More precisely, it evalutes to $1$ if and only if all $g_{\mathbf{v}}$ that are selected evaluate to $1$ and $\un$ otherwise. Let $X_{n-1}$ be the hazard-free decision tree for the $(n-1)$ case where the selector bits are $(s_1,\dotsc,s_{n-1})$ and they select functions $g_{\mathbf{v}}$ where $v_0=0$ and $y_{()}$ and the first $+$ are removed from the functions $g_{\mathbf{v}}$. It is easy to see that this is isomorphic to $f_{n-1}$. The tree $T''_{n-1}$ is obtained by relabeling leaves labeled $0$ in $X_{n-1}$ with $\un$.

    ($s_0=0$ and $y_{()}=1$) The output of $f_n$ must be $1$. So we can output $1$ immediately.

    ($s_0=\un$ and $y_{()}=0$) We know that selector bits must select some $g_{\mathbf{v}}$ where $v_0=0$ and some $g_{\mathbf{w}}$ where $w_0=1$. All such $g_{\mathbf{w}}$ must evaluate to $0$ as $\mathsf{op}_{\mathbf{w}0} = \cdot$. So the value of $f_n$ is $0$ or $\un$ and it is $0$ if and only if all such $g_{\mathbf{v}}$ evaluate to $0$. We can determine this using a decision tree that is isomorphic to that for $f_{n-1}$ where all leaves labeled $1$ are relabeled with $\un$.

    ($s_0 = y_{()} = \un$) In this case, we select both $g_{\mathbf{v}}$ where $\mathsf{op}_{\mathbf{w}0} = +$ and $\mathsf{op}_{\mathbf{w}0} = \cdot$. Since $y_{()} = \un$, the output of $g_{\mathbf{v}}$ where $\mathsf{op}_{\mathbf{w}0} = +$ cannot be $0$ and the output of $g_{\mathbf{v}}$ where $\mathsf{op}_{\mathbf{w}0} = \cdot$ cannot be $1$. Therefore, the output has to be $\un$.

    ($s_0 = \un$ and $y_{()} = 1$) We can deduce that $\mathsf{op}_{\mathbf{w}0}$ for selected $g$ could be $+$ or $\cdot$. Since $y_{()} = 1$. The functions $g$ where this operation is $+$ must evaluate to $1$. Therefore, the remaining decision tree only has to deal with the other half.

     ($s_0 = 1$ and $y_{()} = 0$) Similar to $s_0 = 0$ and $y_{()} = 1$.

     ($s_0 = 1$ and $y_{()} = \un$)  Similar to $s_0 = \un$ and $y_{()} = 0$.

     ($s_0 = 1$ and $y_{()} = 1$)  Similar to $s_0 = 0$ and $y_{()} = 0$.
\end{proof}

\section{Decision Tree Size}

In this section, we prove that despite requiring exponentially more depth in the presence of uncertainty, we can compute the function $\fmux_n$ using a decision tree that is only quadratically larger in size than the Boolean one. However, there are functions that require exponentially larger decision tree size in the presence of uncertainty such as the $\fand$ function. 

\subsection{Decision Tree Size for MUX and AND}\label{app:size}

The size lower bound technique used in this subsection for both the functions involves constructing a set of inputs that must all lead to different leaves and hence any decision tree that computes the hazard-free extension correctly requires size at least as large as the size of this input set. The upper bound is shown by an explicit construction.

\begin{proof}[Proof of Theorem~\ref{thm:muxsize}]

    (of lower bound) Consider any hazard-free decision tree $T$ computing $\fmux_n$. We only consider leaves in $T$ such that on the path from the root to the leaf, all selector bits are queried. We partition this set of leaves into $L_\alpha$ for $\alpha\in\zuon$ such that any leaf in $L_\alpha$ is only reached by inputs where the selector bits have value $\alpha$. Since we only consider leaves where all selector bits are queried, this is a partition of such leaves.

    \begin{lemma}\label{lem:lalpha}
        If $\alpha$ contains exactly $k$ unstable bits, then $|L_\alpha| \geq 2^{k+1}$.
    \end{lemma}

    We complete the proof assuming the above lemma. The calculation is shown below:
    \begin{align*}
        |L| &\geq \sum_{\alpha \in S}|L_\alpha| \\
        &\geq \sum_{k=0}^n 2^{k+1} |\{\alpha \mid \alpha \text{ has exactly $k$ unstable values.}\}|\ \\
        &= \sum_{k=0}^n2^{k+1}{n \choose k}2^{n-k}\\
        &= 2 \cdot 2^n \sum_{k=0}^n {n \choose k} = 2 \cdot 4^n
    \end{align*}

    \begin{proof}(of Lemma~\ref{lem:lalpha})
    Define $I_\alpha$ to be the set of inputs consistent with $\alpha$ on the selector bits such that exactly one data bit indexed by the resolutions of $\alpha$ is $0$ and all other data bits is $1$. These are all $\un$-inputs. Note that $|I_\alpha| = 2^k$. We claim that for any input $z\in I_\alpha$, the tree $T$ must query all selector bits and the data bit (say $x_{(b_0,\dotsc,b_{n-1})}$) with the value $0$ of $z$. Suppose a selector bit $s_j$ is not queried and $\alpha_j=\un$. Consider $z'$ that is the same as $z$ except $s_j=\bnot{b_j}$. The output on $z'$ for $T$ will also be $\un$ when it should be $1$. Suppose $\alpha_j$ is stable. Consider $z'$ such that it differs from $z$ only by $s_j=\bnot{b_j}$. The output on $z'$ is $1$ but $T$ reaches the same leaf for $z$ and $z'$. This establishes that all inputs in $I_\alpha$ are in $L_\alpha$. We also claimed that $x_{(b_0,\dotsc,b_{n-1})}$ is also queried. Indeed, if it is not, we can construct a $z'$ that is exactly the same as $z$ where that data bit is $1$ that will follow the same path. But that input is a $1$-input.

    Notice that two distinct inputs in $I_\alpha$ must reach different leaves in $L_\alpha$. This is because they have unique data bits that are $0$ and they are distinct for distinct inputs. The above claim shows that both these bits have to be queried for those inputs. Therefore, they cannot reach the same leaf. So $|L_\alpha| \geq 2^k$.

    If we consider input set $J_\alpha$ which is similar to $I_\alpha$ except that exactly one indexed (by $\alpha$) data bit is $\un$ and the rest are $1$, we can argue the same for $J_\alpha$ and conclude that there must be $2^k$ inputs that reach distinct leaves. Moreover, these must reach distinct leaves from $I_\alpha$ as no data bit in an input from $I_\alpha$ has value $\un$.\end{proof}

    We remark that the above argument also gives a lower bound on the size of $k$-bit hazard-free decision trees for $\fmux$. Note that Lemma~\ref{lem:lalpha} is already parameterized in terms of the number of unstable bits. By a similar argument, the following lower bound follows:
    any $k$-bit hazard-free decision tree for $\fmux$ must have at least $2^n \left(\sum_{i=0}^k{n \choose i}\right)$.
    
    (of upper bound)We show that $\fmux_n$ has a hazard-free decision tree with $4^{n+1}-3^n$ leaves. First, we query all the selector bits in order. This results in a tree with $3^n$ leaves where each label is naturally associated with $\alpha\in\zuon$ that correspond to the values of the selector bits. For each $\alpha$, we will now proceed to query all the data bits that are indexed by resolutions of $\alpha$. If $\alpha$ has exactly $k$ unstable values, then we will query at most $2^k$ data bits. Let us call this tree (a subtree of the main decision tree) $T_\alpha$. We will have three subtrees for $T_\alpha$, which we call $T_{\alpha,0}$, $T_{\alpha,\un}$, $T_{\alpha,1}$ corresponding to the values of the first data bit. Notice that $T_{\alpha,\un}$ is just a leaf labeled $\un$. From $T_{\alpha,1}$, we query the next data bit. If it is $0$ or $\un$, we immediately output $\un$. Otherwise, we build the tree $T_{\alpha,11}$. Notice that the leaves of $T_\alpha$ correspond exactly to strings:
    \begin{enumerate}
        \item $z\in\zuo^{\leq 2^k}$ such that $z=x\un$ where $x$ is either all $1$ or all $0$ (possibly empty). When $x$ is empty, these are the same string. So there are $2^{k+1}-1$ such strings.
        \item $z\in\zuo^{\leq 2^k}$ such that $z=x0$ where $x$ is non-empty and all ones or $z=x1$ and $x$ is non-empty and all zeroes. There are $2^{k+1}-2$ such strings.
        \item Two more strings $z=1^{2^k}$ or $z=0^{2^k}$.
    \end{enumerate}
    The total is $2^{k+2} - 1$. Summing over all $\alpha$ gives us the size upper bound of

    \begin{align*}
        \sum_{k = 0}^n \binom{n}{k}2^{n-k}(2^{k+2} - 1) = \sum_{k = 0}^n \binom{n}{k}2^{n+2} - \sum_{k = 0}^n \binom{n}{k}2^{n-k} = 4^{n+1} - 3^n.
    \end{align*}

    The first expression here is obtained by summing over $k$ being the number of unstable bits in $\alpha$: there are $\binom{n}{k}$ ways of choosing these bits, $2^{n-k}$ ways of setting the remaining selector bits to values in $\zo$, and the term $2^{k+2}-1$ comes from the argument above.
\end{proof}

\begin{proof}[Proof of Theorem~\ref{thm:andsize}]
    The upper bound is straightforward. To prove the lower bound, let $T$ be a hazard-free decision tree with optimal size computing $\fand_n$. Let $L$ be the set of leaves in $T$. For every leaf $\ell$ of $T$ we will associate a set $V_\ell$ which is the set of pairs $(v,\alpha)$ where $v$ is an input variable and $\alpha \in \zuo$ such that the computation path that leads to the leaf $\ell$ sets variable $v$ to $\alpha$ for all $(v,\alpha) \in V_\ell$. 
    
    Delete all edges labeled $0$ in $T$, remove all sub-trees disconnected from the root, and call the resulting tree $T_{1u}$. We prove that $\size(T_{1u}) \ge 2^n$. We know that there are at least $2^n$ leaves and $2^{n}-1$ internal nodes present in $T_{1u}$. Note that in $T$, every internal node will have a subtree (having at least 1 leaf) where the queried variable is set to 0. Hence there are $2^n -1$ more leaves in $T$ than in $T_{1u}$. Hence we have $\size(T) \ge 2^n + 2^n -1$ and we can use the trivial construction to get a hazard free decision tree of this size.
    
    Similar to the size lower bound proof for $\fmux$, we shall choose a set of inputs (of size $2^n$) which should all definitely reach different leaves. Let $L$ be the set of leaves in $T_{1u}$. For every leaf $\ell$ of $T_{1u}$, consider the set of inputs $I = \{x \mid x \in \{1, u\}^n \}$ which does not have $0$. We claim that if $x \in I$ and $T$ reaches the leaf $\ell$ on input $x$, then $V_{\ell}$ has all variables of $x$. Assume this claim, the argument is complete since we know that every $x \in I$ goes to a different leaf and $|L| \ge 2^n$ thereby proving the proposition.
    
    To prove the claim, suppose $x_i$ is not contained in $V_\ell$. We split the proof into two cases. If $T$ outputs $1$ or $u$, consider an input $x'$, such that $x'_i = 0$ and $x'_j = x_j \forall j \neq i$, note that this reaches $\ell$ as well and $\fand(x') = 0 \neq T_{1u}(x') = T(x')$, a contradiction. If $T$ outputs $0$, then observe that $\fand$ evaluates to $1$ on $x$, a contradiction.
\end{proof}

\subsection{Almost Optimal Decision Tree Size Lower Bound}\label{app:cm}

\begin{theorem}\label{thm:pisize}
    Let $f$ be any function. Then, $\size_u(f) \geq m+M$ where $m$ is the number of prime implicants of $f$ and $M$ is the number of prime implicates of $f$.
\end{theorem}

\begin{proof}
    We define an input $x_P$ that corresponds to a prime implicant as follows: All variables that occur positively in $P$ are set to 1, those that occur in a negated form in $P$ are set to 0, and the remaining are set to $\un$. Let $X_P$ be the set of all $x_P$ where $P$ is a prime implicant. Since no two of these inputs are the same we know that $|X_P| = m$.

    The theorem now follows from the proof of this claim.
    \begin{claim}
        Suppose $T$ is a decision tree computing $\uf$. Every $x_P \in \mathcal{X_P}$ reaches a leaf $L$ such that every index that has a value of $0$ or $1$ in $x_P$ is queried on the path from the root to $L$.
    \end{claim}

    \begin{proof}
    Suppose for the sake of contradiction, $x_P$ reach the leaf $L$, with $V(L)$ (the set of variables it has queried and the value that each has been set to) and a variable in $P$ is not queried. We know that the values of the indices in $V(L)$ have to match $x_P$ i.e ${x_Q}_i = b_i \forall (i, b_i) \in V(L)$. Consider the input $x_L$ such that ${x_L}_i = b_i \forall (i, b_i) \in V(L)$ and the remaining bits are set to $u$. Note that there exists an index $j$ where ${x_P}_j = b_j$ and ${x_L}_j = u$, and hence we know that on flipping certain $u$'s in $x_L$(at least) to Boolean bits we get $x_P$. Note that $x_L$ corresponds to an implicant $L'$ (exactly the boolean bits in $x_L$ are considered) which covers the implicant $P$ which is a contradiction as $P$ is a prime implicant.
    \end{proof}

    The theorem follows.
\end{proof}

\begin{definition}
    Consider the function $\mathsf{cm}_n(x_1, \dotsc, x_n)$ define as follows: The output is $1$ if and only if the number of $0$s and the number of $1$s are in the range $[n/3, 2n/3]$.
\end{definition}

It is easy to see that the function $\mathsf{cm}_n$ has $\binom{n}{n/3}\binom{2n/3}{n/3}$ prime implicants. Using Theorem~\ref{thm:pisize}, we obtain an almost optimal decision tree size lower bound for hazard-free extensions of Boolean functions.

\subsection{Constructing Decision Trees for Hazard-free Extensions}\label{app:dttoudt}

In this section, we present a general construction of decision trees of hazard-free extensions of functions from a decision tree of the function. The main idea is to construct the $\un$ subtree of the root node in the new decision tree from the $0$ and $1$ subtrees that have been constructed recursively. The $\un$ subtree construction can be viewed as a product construction where we replace each leaf in a copy of the $0$ subtree with a modified copy of the $1$ subtree. This product construction works because for every input that reaches the $\un$ subtree, the output value can determined by the output values when that bit is $0$ and when it is $1$, which is exactly what the $0$ and $1$ subtrees compute. 

\begin{proof}[Proof of Theorem~\ref{thm:dttoudt}]

    (first inequality) Given a decision tree $T'$ for $\uf$, consider $T$ obtained by removing all subtrees from $T'$ that can only be reached by following an edge labelled $\un$. That is, traverse over all root-to-leaf paths in $T'$, and truncate all subtrees that are encountered the moment an edge labelled $\un$ is read. It is easy to see that $T$ is a decision tree computing $f$, and hence $\size(T) \ge \size(f)$. Since $T'$ has at least one more leaf for each internal node of $T$ (corresponding the query response of $\un$ at that node) and since the number of internal nodes of a tree is one less than its size, we have that $\size(T') \ge \size(T) + \size(T) - 1 = 2\size(T)-1 \ge 2\size(f)-1$. This inequality is tight as witnessed by the parity function on $n$ inputs.
    
    (second inequality) We describe how to construct a decision tree $T'$ computing $\uf$ from a Boolean decision tree $T$ computing $f$.
    
    (Construction of $T'$) Let $v$ be a node querying a variable $x_v$ in the original tree $T$. Let $T_{v, 0}$ (resp. $T_{v, 1}$) denote subtrees of $T$ when $x_v = 0$ (resp. $x_v =1$). Let $v_0$ and $v_1$ be the children of node $v$ in the decision tree $T$ through a bottom-up inductive approach where we add $\un$ query subtrees to every node of $T$.
    
    Let us assume inductively that when a $\un$ query subtree is added to any node of $T$, all of the children of the node already have their $\un$ query subtrees attached. The base case would be at the leaves where you do not have to add anything since there is no variable to be queried.
    
    The inductive step at a node $v$ would be as follows: by induction we assume to have converted the decision tree $T_{v, b}$ rooted at $v_b$ to $T_{v, b}'$ for $b \in \{0, 1\}$. Let $\rho(v)$ be the set of tuples of variables and the values that they are set to in the path from the root to the node $v$. For $b \in \{0,1\}$, and any leaf $\ell \in T_{v, b}'$, let $\rho_{v,b}(\ell)$, the set of tuples of variables and the values that they are set to in the path from $v_b$ to the leaf $\ell$.
    
    We will now construct a subtree $T'_{v, \un}$ of $v$ for the case $x_v = \un$. The subtree $T'_{v, \un}$ is a copy of $T_{v, 0}'$ with the leaves labelled $0$ and $1$ of $T_{v, 0}'$ replaced with copies of $T_{v, 1}'$ appropriately modified as follows: For $b \in \{0,1\}$, let $\ell$ be the leaf in the new copy of $T_{v, 0}'$ labelled with $b$. Replace the leaf $\ell$ with a copy of $T'_{v, 1}$ with the following modifications:
    \begin{enumerate}
        \item Remove any branches that are inconsistent with assignments in $\rho_{v, 0}(\ell)$.
        \item Replace the $\overline{b}$-leaves with $\un$.
    \end{enumerate}
    
    (Correctness) We show that for any input $x \in \zuon$ that satisfies $\rho(v)$, the subtree constructed above ensures that the leaf that the subtree reaches has the value $f(x)$. The proof is by induction on the height of $v$.
    
    \begin{itemize}
        \item (Base case) The base case is when $v$ is a leaf. In this case, the construction is trivially correct.
    
        \item {(Inductive case)} Consider a non-leaf $v$, For a child $w$ of $v$, we can now assume that any input $x \in \zuon$ that satisfies $\rho(w)$ reaches the value $f(x)$ in the subtree that has been constructed at $w$ (including $T_{w, 0}', T_{w, 1}', T_{w, \un}'$). Any $x$ consistent with $\rho(v)$, and $x_v$ is $0$ or $1$, the tree is lead to the subtrees rooted at $v_0$ or $v_1$, which gives the correct value by induction hypothesis.
    
        Now we consider the case of $x$ consistent with $\rho(v)$ and $x_v = u$. Let $x_0$, and $x_1$ be the inputs obtained by replacing the $x_v$ with $0$ and $1$ respectively. We know that:
        \[
            f(x) = \begin{cases}
                b & \textrm{ if $f(x_0) = f(x_1) = b$} \\
                u & \textrm{ otherwise} 
            \end{cases}
        \]
        
        We now have two cases based on value of $f(x_0)$.
        \begin{itemize}
            \item $f(x_0) = b$ where $b \in \{0,1\}$: We know that $x_0$ leads the tree to a leaf $\ell_0$, labelled $b$, via the assignment $\rho(v) \cup \rho_0(\ell) \cup \{(x_v,0)\}$ in the tree $T'_{v, 0}$. The input $x_1$ leads the tree to a leaf $\ell_1$ in $T_{v, 1}'$ and the corresponding path is represented by the assignment $\rho(v) \cup \rho_1(\ell_1) \cup \{(x_v,1)\}$. 
            
            Consider $x$ in $T'$. It reaches the node $v$ in $T$ (since it is consistent with $\rho$) and it further reaches $T_{v, \un}'$ since $x_v=\un$. From there, the input $x$ will follow the path to the node corresponding to the assignment $\rho(v) \cup \{(x_v,u)\} \cup \rho_0(\ell_0)$ since $T_{v, \un}'$'s first part is same as $T_{v,0}'$ by the construction. Further, since the leaf ($\ell_0$) in the copy of $T_{v,0}'$ has been replaced with a modification of $T_{v,1}'$ restricted to $\rho_0$, $x$ follows the path to the leaf indexed by $\rho(v) \cup \{(x_v,u)\} \cup \rho_0(\ell_0) \cup \rho_1(\ell_1)$ and reaches the new leaf in the modified copy of $T_{v, 1}'$. In this case, the output would be 
            \[
            T'(x) = \begin{cases}
                b & \textrm{ if $f(x_1) = b$} \\
                u & \textrm{ otherwise} 
            \end{cases}
            \] as desired.
            \item $f(x_0) = u$: We know from induction hypothesis that $x_0$ follows the path to a leaf ($\ell$) labelled $u$ via the assignment $\rho(v) \cup \{(x_v,0)\} \cup \rho_0(\ell)$. Therefore, the input $x$ firstly reaches $T_{v,\un}'$ and since this subtree is a copy of $T_{v, 0}'$, the input $x$ will reach a leaf labelled $\un$ by following a similar path as $x_0$.
        \end{itemize}
        
    \end{itemize}

    (Size of $T'$) An easy upper bound on $\size(T'_{v, \un})$ is $\size(T'_{v, 0})\size(T'_{v, 1})$. we have:
    \begin{eqnarray*}
    \size(T_v') & \le & \size(T_{v, 0}')+\size(T_{v, 1}')+\size(T_{v, \un}') \\
    & \le & \size(T_{v, 0}')+\size(T_{v, 1}')+\size(T'_{v, 0})~\size(T'_{v, 1}) \\
    & \le & (\size(T'_{v, 0})+1)(\size(T'_{v, 1})+1)-1
    \end{eqnarray*}
    Hence for any node, we have the following recursive relation for the size of the tree rooted at that node. $\size(T')+1 \le (\size(T_{r, 0}')+1)(\size(T_{r, 1}')+1)$. Notice that the two terms on the right-hand side has the same form as the expression on the left-hand side. Since the construction is recursive, we can keep recursively expanding all the terms on the right-hand side until we reach the leaves. At any leaf, this expression is just $(1+1)$ and there will be $\size(T)$ many of them. So $\size(T') + 1 \leq 2^{\size(T)}$ as desired.
\end{proof}

\section{Constructing Decision Trees for Limited Uncertainty}\label{app:kbit}

We also prove that the complexity increases only gradually along with the increase in the amount of uncertainty in the inputs. More specifically, we prove that if the inputs are guaranteed to have at most $k$ bits with value $\un$, without any guarantee on where they occur, the exponential blow-up in query complexity is contained to the parameter $k$. 
This proof also makes use of the product construction mentioned above.

\begin{proof}[Proof of Theorem~\ref{thm:kbitsize}]
    We construct a sequence of trees $T^0, T^1, T^2, ..., T^{k}$ such that $T^i$ is a that can correctly handle inputs with at most $i$ positions of value $\un$. We denote by $T^i_{v,b}$ for $b\in\zuo$ the tree obtained from the $b$-subtree of node $v$ in $T$ by iterating this process $i$ times. We define $T^0 := T$ for all trees $T$.

    {(Constructing $T^{k+1}$)} We construct a subtree $T^{k+1}_{v,\un}$ to link to the node $v$ when the query for variable $x_v$ returns $u$. We will build the subtree $T^{k+1}_{v, \un}$ from a copy of $T^k_{v, 0}$ by replacing the leaves labelled $0$ and $1$, with copies of $T^k_{v, 1}$ appropriately modified as follows: for $b \in \{0,1\}$, let $\ell$ be the leaf in the new copy of $T^k_{v, 0}$, labelled with $b$. Replace the leaf $\ell$ with a copy of $T^k_{v, 1}$ with the following modifications:
    \begin{enumerate}
        \item Remove any branches that are inconsistent with assignments in $\rho_0(\ell)$.
        \item Replace the $\overline{b}$-leaves with $u$.
    \end{enumerate}

    (Correctness) The tree $T^0$ can correctly handle inputs with no $\un$s. We now prove the inductive step is correct.
    
    Consider any input $x \in \zuon$ such that the number of $u's$ is $x$ is at most $k+1$. When the $T^{k+1}$ is instantiated with $x$ let $v$ be the first node that is set to $x_v = u$. If such a node is not encountered then the $x$ reaches a node of the original tree and it is easy to prove that all of it's resolutions reach that in $T$ and hence $T^{k+1}$ outputs $f(x)$. Otherwise let $x_0$ and $x_1$ be inputs obtained by replacing $x_v$ in $x$ with 0 and 1 respectively. We know that 
    \[
        f(x) = \begin{cases}
            b & \textrm{ if $f(x_0) = f(x_1) = b$} \\
            u & \textrm{ otherwise} 
        \end{cases}
    \]

    Assuming the correctness of $T^k$ and since the number of $u$ values in both $x_0$ and $x_1$ is at most $k$, the proof by cases in the correctness of the  decision tree construction of Theorem~\ref{thm:dttoudt} can be used by replacing $T'_{v, 0}$ and $T'_{v, 1}$ by $T^k_{v, 0}$ and $T^k_{v, 1}$ respectively.

    (Size) Note that we $L(T^{k+1}_{v, \un})$ can be upper bounded by $L(T^k_{v, 0})L(T^k_{v, 1})$ and we have 
    
    $$L(T^{k+1}) \le L(T) + \sum_{n \in \mathcal{I}(T)}L(T^{k+1}_{n, \un})$$
    where $\mathcal{I}(T)$ is the set of internal nodes of $T$ and $T^{k+1}_{n, \un}$ is the $x_N = \un$ subtree added to $T$ in order to construct $T^{k+1}$.

    \begin{eqnarray*}
    L(T^{k+1}) & \le & L(T) + \sum_{n \in \mathcal{I}(T)}L(T^{k}_{n, 0})L(T^{k}_{n, 1}) \\
    & \le & L(T) + \sum_{n \in \mathcal{I}(T)}L(T^{k}_{n, 0}) L(T^{k}_{n, 1}) \\
    & \le & L(T) + \sum_{n \in \mathcal{I}(T)}{(L(T^k) - 1)}^2 \\
    & \le & L(T) + \sum_{n \in \mathcal{I}(T)}({L(T^k)}^2 - 1) \\
    & \le & L(T){L(T^k)}^2
    \end{eqnarray*}

    On recursively expanding $L(T^k)^2$, we have $L(T^{k+1}) \le L(T)^{2^{k+2} - 1}$.

    (Depth) Let the leaf with the maximum depth in $T^{k+1}$ be $\ell$ with $\rho(\ell)$ its set of tuples of variables and values that they are set to in the path from the root to the leaf. Suppose $\rho(\ell)$ does not have any $\un$ that its variables are set to we are done, since then $d(T^{k+1}) = d(T)$. Else suppose $v$ is the first node whose variable is set to $\un$ in the path from root to $\ell$. We know that 
    \begin{eqnarray*}
    d(T^{k+1}) & \le & d(v) + d(T^{k+1}_{v, \un}) \\ & \le & d(v) + d(T^{k}_{v, 0}) + d(T^{k}_{v, 1}) \\
    & \le & 2 d(T^{k})
    \end{eqnarray*}

    Recursively expanding gives $d(T^{k+1}) \le 2^{k+1}d(T^0)$. This completes the proof.
\end{proof}

\section{Learning Hazard-free Extensions with Low Sensitivity}\label{app:learning-sensitivity}

The Hamming distance between two inputs $x, y$ in the 3 valued logic is defined as: $D(x, y) = |\{i \mid x_i \ne y_i \}|$. A Hamming sphere of radius $r$ centered at $x$ is defined as $S(x, r) = \{y \mid D(x, y) = r\}$. A Hamming ball of radius $r$ centered at $x$ is defined as $B(x, r) = \{y \mid D(x, y) \leq r\}$.

\begin{proof}
    For $f : \zuon \rightarrow \zuo$ which is a hazard free extension of a Boolean function we define {\em neighborhood} of $x$ as: $N(x) = \{y \mid \exists i~ x_{i} \neq y_{i}, \textrm{ and } \forall j \in [n] \setminus \{i\}, x_{j} = y_{j} \}$. Similarly we define {\em i-neighborhood} $N_{i}(x) = \{y \mid  x_{i} \neq y_{i}, x_{j} = y_{j} \forall j \in \{1, 2, ..., n\} \setminus \{i\} \}$. Note that for any $x \in \zuon$ and $i \in [n]$, we have that $|N(x)| = 2n$ and $|N_{i}(x)| = 2$.

    We first prove the following Lemma.
    \begin{lemma}\label{lem:neighbors}
    For any function $f$ such that $\sen_u(f) = s$, if $S \subseteq N(x)$ where $|S| \ge 4s+1$ then $f(x) =\fplu_{y \in S}(f(y))$ where $\fplu$ outputs the most frequently occurring element in $\zuo$ in the input to it.
    \end{lemma}
    
    \begin{proof}
        We shall prove this lemma by showing that a set of neighbors with certain properties specifies $f$. Consider $S \subseteq [n]$ such that $|S| = k \ge 2s+1$. Define $N_{S}(x) = \{y_{i_1}, y_{i_2}, ..., y_{i_k}\}$ such that $\forall i_j \in S$, $y_{i_j} \in N_{i_j}(x)$, we claim that $f(x) = \fplu_{y \in N_{S}}(f(y))$.
    
        To see this claim, note that, in $N_{S}(x)$, $x$ has neighbors in at least $2s+1$ different indices. Let the frequency of occurrences of $0$, $1$ and $u$ over these neighbors be $f_{v_1}, f_{v_2}$, and $f_{v_3}$. Let $f(x) = v_{i}$, we know that it will have the remaining $f_{v_{j}}+f_{v_{k}}$ to be the number of sensitive indices, which can at most be $s$. 
        Since $x$ has at least $2s+1$ neighbors in $N_S(x)$, $v_{i}$ must occur at least $s+1$ times among the neighbors in $N_S(x)$. Hence the claim follows.

        We will now use the claim to prove Lemma~\ref{lem:neighbors}. Suppose we have $S \subseteq N(x)$ where $|S| \ge 4s+1$, then by the pigeonhole principle, $S$ contains elements from at least $2s+1$ different $N_{i}$'s. If not, as we know $|N_{i}(x)| = 2$, $N_{i}(x) \cap N_{j}(x) = \phi$ and $N(x) = N_{1}(x) \cup N_{2}(x) .... \cup N_{n}(x)$. If $S \subseteq N(x)$ and $S$ does contains elements from at most $2s$ different $N_{i}$'s then $|S| \le 4s$.

        Hence from the $2s+1$ different $N_{i}$'s we get the index set $S' \subseteq [n]$ and neighbors in each of these index sets $N_{S'}(x) = \{y_{i_1}, y_{i_2}, \ldots, y_{i_k}\}$ such that $|S'| \ge 2s+1$ and $y_{s_i} \in N_{s_{i}}(x) \forall s_{i} \in S'$. We know that $f(x) = \fplu_{y \in N_{S'}}(f(y))$. Suppose $f(x) = v_{i}$, we will now prove that $v_{i} = \fplu_{y \in S}(f(y))$. Let the other values be $v_{j}, v_{k}$. We claim that $f_{v_j} + f_{v_k} \leq 2s$. Suppose not. Then the number of indices that are sensitive are at least $s+1$ by pigeon hole principle, which is a contradiction. Hence $f_{v_j} + f_{v_k} \leq 2s$, $v_{i}$ occurs at least $2s+1$ times in $S$ and hence $v_{i} = \fplu_{y \in S}(f(y))$. This completes the proof of the lemma~\ref{lem:neighbors}.
    \end{proof}

    Now we complete the proof of the theorem. Suppose the values of a hazard-free extension $f$ with $\sen_u(f) \le s$ is known over a Hamming ball $B(x,4s)$ centered at input $x$ and of radius $4s$ in $\zuon$. We will prove through induction that if the values in any ball $B(x,r)$ where $r \ge 4s$ is known, the values at $S(x,r+1)$ is fixed.

    For the induction step, suppose $B(x, r)$ where $r \ge 4s$ is fixed and $y \in S(x, r+1)$. We claim that $N(y) \cap S(x, r) = r+1$. Let $\mathsf{Diff}(x, y) = \{i \mid x_i \neq y_i\}$. Note that for every $i \in \mathsf{Diff}(x, y)$ such that $x_{i} \neq y_{i}$, we have a corresponding $z^i$ such that $z^{i}_{j} = y_{j} \forall j \neq i$ and $z^{i}_{i} = x_{i}$. $z^i \in S(x, r)$ and $z^i \in N(y)$. It is easy to see that none of the other neighbors of $y$ lie in $S(x,r)$.

    Hence from Lemma~\ref{lem:neighbors} we can fix $f(y)$ since the value of $f$ at at least $4s+1$ neighbors are known. Similarly, we can fix all $f(y)$ for every $y \in S(x, r+1)$. This completes the proof.
\end{proof}

\section{Discussion and Open Problems}

In this paper we initiated a study of query complexity of Boolean functions in the presence of uncertainty, by considering a natural generalization of Kleene's strong logic of indeterminacy on three variables.

We showed that an analogue of the celebrated sensitivity theorem~\cite{Hua19} holds in the presence of uncertainty too. While Huang showed a fourth-power relationship between sensitivity and block sensitivity in the Boolean world, we were able to show these measures are linearly related in the presence of uncertainty. The proof of sensitivity theorem in our setting is considerably different and easier from the proof of sensitivity theorem in the Boolean world. We can parameterize $\un$-sensitivity and $\un$-query complexity by restricting our attention to inputs that have at most $k$ unstable bits. The setting $k = 0$ gives us the Boolean sensitivity theorem and the setting $k = n$, our sensitivity theorem. Can we unify these two proofs using this parameterization? That is, is there a single proof for the sensitivity theorem that works for all $k$?

We showed using $\un$-analogues of block sensitivity, sensitivity, and certificate complexity that for all Boolean functions $f$, its deterministic, randomized, and quantum $\un$-query complexities are polynomially related to each other. An interesting research direction would be to determine the tightest possible separations between all of these measures. It is interesting to note that our quadratic relationship between deterministic and randomized $\un$-query complexity improves upon the best-known cubic relationship in the usual query models. Moreover, our quartic deterministic-quantum relationship matches the best-known relationship in the Boolean world~\cite{ABKRT21}. More generally, it would be interesting to see best known relationships between combinatorial measures of Boolean functions in this model, and see how they compare to the usual query model (see, for instance,~\cite[Table~1]{ABKRT21}). A linear relationship between deterministic and randomized query complexities in the presence of uncertainty remains open, but a quadratic deterministic-quantum separation follows from Theorem~\ref{thm:muxdepth} or from the OR function (via Grover's search algorithm~\cite{Gro96}).

While we studied an important extension of Boolean functions to a specific three-valued logic that has been extensively studied in various contexts, an interesting future direction is to consider query complexities of Boolean functions on other interesting logics. Our definition of $\uf$ dictates that $\uf(x) = b \in \zo$ iff $f(y) = b$ for all $y \in \res(x)$, and $\uf(x) = \un$ otherwise. A natural variant is to define a 0-1 valued function that outputs $b \in \zo$ iff majority of $f(y)$ equals $b$ over all $y \in \res(x)$. It is not hard to show that the complexity of this variant of $\fmux_n$ is bounded from below by the usual query complexity of Majority on $2^n$ variables, which is $\Omega(2^n)$ in the deterministic, randomized, and quantum query models.

\bibliographystyle{alpha}
\bibliography{refs.bib}

\appendix

\section{Bounding CREW-PRAM Time in the Presence of Uncertainty}\label{app:cooknisanun}

The folklore "pointer-doubling" technique used by Nisan \cite{Nisan91} to prove that CREW-PRAM time complexity is at most $\log_2(\depth(f))$ can be adapted easily to show that CREW-PRAM time complexity of $\uf$ is at most $\log_2(\depth_\un(f))$.

For completeness, we show Cook, Dwork, and Reischuk's \cite{CDR86} lower bound on the time required by a parallel random access machines without concurrent writes in terms of sensitivity of the function can also be extended. They proved that any PRAM (as defined below) computing a function $f$ requires at least $\log_bs(f)$ steps where $b = \frac{1}{2}(5 + \sqrt{21})$. In this section we shall extend this result to say that any CREW-PRAM require at least $\log_b{\sen_\un(f)}$ time to compute $\uf$.

\begin{definition}[PRAM]
    A PRAM consists of a set $\Pi  = \{P(1), P(2), \cdots\}$ of processors, a set $\Gamma = \{M(1), M(2), \}$ of cells, an alphabet $\Sigma$, a number $n$ of inputs, and an execution time $T$. Each processor $P(i)$ consists of a state set $Q_i$ (any of the sets $\Pi, \Gamma ,\Sigma, Q_i$ may be infinite), and functions $\rho_i: Q_i\rightarrow \mathbb{N}^{+}, \tau_i: Q_i \rightarrow \mathbb{N}, \sigma_i: Q_i \rightarrow \Sigma $ and $\delta_i: Q_i \times \Sigma \rightarrow Q_i$. Here for any $q \in Q_i$, $\rho_i(q)$ decides the memory index to be read, $\delta_i$ represents the state transition function that changes the state of $P_i$ based on the previous state and the  alphabet read, $\tau_i(q)$ represents the memory index to be written into next and $\sigma_i(q)$ is the alphabet that will be written in that memory index next. ($\tau_i(q)$ indicates that no cell is written into)

    At each time $t= 0, 1,. ., T$ each processor $P(i)$ is in a state $q \in Q_i$ and each cell $M(i)$ contains a symbol $s_{i}^t \in \Sigma$. At time $t = 0$, cells $M(1),..., M(n)$ contain the inputs $X_1,\cdots,X_n$. That is, $s_{i}^t  = X_i \forall i \in [n]$, and $s_{i}^0 = b_0 \forall i>n$, where $b_0$ is some distinguished (blank) member of $\Sigma$. All processors are initially in the distinguished state $q_0$  i.e. $q_{i}^0 = q_0 \forall i$. In general,
    $$q_{i}^{t+1} = \delta_i(q_{i}^t, s_{j}^t)$$
    for all processors $P(I)$ where the index read is $j = \rho_i(q_i^t)$
    
    \[
        s_{k}^{t+1} = \begin{cases}
            \sigma_i(q_{i}^{t+1}) & \textrm{ if $k = \tau_i(q_{i}^{t+1})$} \\
            s_{k}^t & \textrm{ otherwise} 
        \end{cases}
    \]

    It is a condition of correctness of the PRAM (expressing the constraint of not allowing concurrent writes) that for $t = 0, 1,. , T- 1$, $\tau_{j}(q_{j}^{t+1})$ and $\tau_{k}(q_{k}^{t+1})$  are either both zero or distinct for all $j \neq k$. The value $f(X1,\cdots, Xn)$ of the function $f$ computed by the PRAM is the contents $s_{1}^T$ of cell $M(1)$ at time $T$.
\end{definition}

For an input string $I \in \zuon$, we denote by $I(i)$ that 2 strings which differs from I exactly at position $i$. We use $I(i, b_i)$ to denote the string that differs from $I$ exactly at position $i$ by the bit $b_i$.

\begin{definition}
    An index $i$ affects a processor $P$ (and respectively, a cell $M$) at time $t$ with $I$ iff the state of $P$ (and respectively the content of $M$) at input configuration $I$ differs from the state of $P$ (and respectively the content of $M$) for any of the inputs configurations in $I(i)$.

    Index-bit pair $(i, b_i)$ affects  a processor $P$ (and respectively, a cell $M$) at time $t$ with $I$ iff the state of $P$ (and respectively the content of $M$) at input configuration $I$ differs from the state of $P$ (and respectively the content of $M$) for $I(i, b_i)$.    
\end{definition}

\begin{theorem}
    If $f: \{0, 1, \un\}^n \rightarrow \{0, 1, \un\}$ has $\sen_\un(f)$ as its sensitivity then any PRAM computing it will require at least $\log_b \sen_un(f)$ time to compute it.
\end{theorem}

\begin{proof}
    Let $K(P, t, i)$ (respectively, $L(M, t, i)$)  be the set of input indices which affect processor $P$ (respectively,  memory cell $M$) at with $I$. Let $K_t, L_t$, satisfy the recurrence equations:
    \begin{enumerate}
        \item $K_0 = 0$
        \item $L_0 = 1$
        \item $K_{t+1} = K_{t} + L_{t}$
        \item $L_{t+1} = L_{t} + 3K_{t+1}$
    \end{enumerate}

    As proven in \cite{CDR86}, the solution to these equations is $$K_{t} = \frac{b^t}{\sqrt{21}} - \frac{\Bar{b}^t}{\sqrt{21}}, L_{t} = \frac{3 + \sqrt{21}}{2\sqrt{21}}b^t + \frac{-3 + \sqrt{21}}{2\sqrt{21}}\Bar{b}^t$$
    where $b = \frac{1}{2}(5+\sqrt{21}), \Bar{b} = \frac{1}{2}(5 - \sqrt{21})$. Consider the input $I$ with the maximum sensitivity, we know that at $T$ all the $\sen_\un(f)$ bits affect the content of $M(1)$ and hence $|L(M(1), I, T)| = \sen_\un(f)$. By the below lemma, we can say that $$\sen_\un(f) \le L_{T} \le b^T \implies T \ge \log_b \sen_un(f)$$
    
    \begin{lemma}
        $|K(P, I, t)| \le K_{t}$ and $|L(M, I , t)| \le L_{t}$ for all $P, M, t$ and $I$
    \end{lemma}
    We now prove the above lemma by induction as follows:
    \begin{description}
        \item[Hypothesis:] $|K(P, I, t)| \le K_{t}$ and $|L(M, I , t)| \le L_{t}$ for all $P, M, I$ at time $t$
        \item [Base case:] At $t = 0$, we know that $K(P, I, 0) = 0$ and $L(M, I, 0) = 1$

        \item [Induction step:]
        \begin{enumerate}
            \item Consider $K(P, I, t+1)$. Any index $i$ that affects $P$ at time $t+1$ would have either affected $P$ at time $t$ or it would change the cell $M(\rho_i(q_{i}^t))$ that it would've read from (it is the same cell $P$ would've read from for input configuration $I$ at time $t$). Hence we know that $K(P, I, t+1) \subseteq K(P, I, t) \cup L(M, I, t)$. This satisfies $|K(P, I, t+1)| \le |K(P, I, t)|+ |L(M, I, t)| \le K_t + L_t = K_{t+1}$

            \item Now let us consider $L(M, I, t+1)$. We have the following 2 cases:

            \begin{description}
                \item[Case 1:] Suppose that for the input configuration $I$, some processor $P$ writes into $M$. Any index $i$ that affects $M$ at $t+1$ would definitely affect $P$ at time $t+1$ (if not then $P$ would've written the same alphabet into $M$). Hence we have $L(M, I, t+1) \subseteq K(P, I, t+1)$ and hence $|L(M, I, t+1)| \le K_{t+1} \le L_{t+1}$ since $K_t, L_t$ are non negative for all $t$.
                \item[Case 2:] If no processor had written into $M$, then index $i$ (say by the change to a bit $b_i$) would affect $M$ at $t+1$ it it had either affected it at $t$ or if it causes some processor to write into $M$ at time $t$. Say $Y(M, I, t+1)$ denotes the number of indices that when changed to any of the other bits makes some processor $Q$ write into $M$. Then we know that $L(M, I, t+1) \subseteq L(M, I, t) \cup Y(M, I, t+1)$. We shall upper bound $|Y|$ using the following lemma by $K_{t+1}$. Hence we have $|L(M, I, t+1)| \le |L(M, I, t)| + Y(M, I, t+1)| \le L_{t} + 3K_{t+1} = L_{t+1}$
            \end{description}
        \end{enumerate} 
    \end{description}

    \begin{lemma}
        Suppose $M$ is not written into by any processor at $t+1$ for the input configuration $I$ then $Y(M, I, t+1) \le 3K_{t+1}$
    \end{lemma}
    Consider $Y(M, I, t+1) = \{(u_1, b_1), ..., (u_r, b_r)\}$ where any $I(u_i, b_i)$ is the input where the index $u_i$ (all the indices are unique) has its bit replaced by $b_i$ and this makes $Q_i$ write into $M$. If there are multiple possible $b_i$'s for a $u_i$ we pick any one. 

    \begin{observation}
        Note that if $Q_i \neq Q_j$ then either $(u_i, b_i)$ affects $Q_j$ at $t+1$ for input $I(u_j, b_j)$ or $(u_j, b_j)$ affects $Q_i$ at $t+1$ for input $I(u_i, b_i)$
    \end{observation}
    This happens simply because if not, for the input $I(u_i, b_i)(u_j, b_j)$ both $Q_i$ and $Q_j$ are going to write to $M$.

    In order to bound $r$ we shall create the bipartite graph with vertices $(u_1, b_1), \cdots, (u_r, b_r)$ and $(v_1, b_1), \cdots, (v_r, b_r)$. There is an edge between $(u_i, b_i)$ and $(v_j, b_j)$ if and only if $(u_i, b_i)$ affects $Q_j$ for $I(v_j, b_j)$ at $t+1$. Consider the in-degree of $(v_j, b_j)$, note that it can only be at most $K(Q_j, I(v_j, b_j), t+1)$ since all the $(u_i, b_i)$ pairs affecting it have unique indices. Hence we have $$e \le r K(Q_j, I(v_j, b_j), t+1) \le rK_{t+1}$$
    We shall use the observation above to lower bound $e$ by the number of pairs $((u_i, b_i), (u_j, b_j))$ that lead to different processors writing to $M$. For a fixed $(u_i, b_i)$ there can be at most $K(Q_i, I, t+1)-1$ other pairs that affect $Q_{i}$. Hence there are at least $r-|K_(Q_i, I, t+1)|$ pairs that will cause some other processor to write into $M$ (note that $Q_i$ has to be affected in order to write into $M$). We have

    $$\frac{1}{2}r(r - K_{t+1}) \le \frac{1}{2}r(r - |K_(Q_i, I, t+1)|) \le e$$

    Hence we have $\frac{1}{2}r(r - K_{t+1}) \le rK_{t+1} \implies r \le 3k_{t+1}$
\end{proof}

\section{Alternative Definitions for Sensitivity of Hazard-free Extensions}\label{app:alternative-sen}

We consider two other natural definitions of sensitivity for hazard free extensions and prove that they are linearly equivalent to the definition of sensitivity for hazard-free extensions used in this paper.

\newcommand{\stabs}{\mathsf{stbs}}
In the following definition, we consider a bit to be sensitive if for some input, we can change the value of that bit from $0$ or $1$ to $\un$ and make the output of the function $\un$.
\begin{definition}[Stable Sensitivity]
Let $f$ be the hazard-free extension of some $n$-input Boolean function. For an $x\in \{0, \un, 1\}^n$ such that $f(x)$ is stable, we define the stable sensitivity $\stabs(f,x)$ of $f$ at $x$ as the number of bits that can be made unstable (individually) so that $f(x)$ is unstable. 
$$\stabs(f) = \max_{x : \textrm{$f(x)$ is stable}} \stabs(f,x)$$
\end{definition}

\newcommand{\slys}{\mathsf{stys}}
In the following definition, we consider a bit to be sensitive for some input if switching the value from a certain to uncertain value changes the output from certain to uncertain or vice-versa.
 For $B \subseteq [n], x \in \{0,1,\un\}^n$, we define $x \oplus B = \{ y \in \{0,1,\un\}^n \mid \forall i \in B, x_i = \un \iff y_i \neq \un$.
\begin{definition}[Stability Sensitivity]
    Let $f$ be the hazard-free extension of some $n$-input Boolean function. For an $x\in \{0, \un, 1\}^n$, the Stability Sensitivity at the input $x$, denoted by $\slys(f,x)$, is defined as:
    $$\slys(f,x) = \left| \left\{ i : \begin{array}{l} \exists y \in x \oplus \{i\},     f(x) = \un \iff f(y) \neq \un
    \end{array}
    \right\} \right|$$
    The Stability Sensitivity of $f$, $\slys(f)$ is then defined as :
    $$\slys(f) = \max_{x \in \{0,1,\un\}^n} \slys(f,x)$$
\end{definition}

We shall prove that these measures of sensitivity $\stabs(f), \slys(f)$ are at polynomially related to $\sen(f)$ for any hazard free extension $f$. In fact, similar to the relation between $\sen_\un(f), \bsen_\un(f)$ and $\cc_\un(f)$ they are also away from $\sen(f)$ by a factor of at most 2. We use the notation $\stabs_\un^{(b)}, \slys_\un^{(b)}$ where $b\in \{0, 1\}$ to denote sensitivities for inputs that yield a $b$ output and $\slys_\un^{(u)}$ to denote Stability Sensitivity for inputs that yield a $\un$ output.

\begin{claim}
    Let $f$ be the hazard -free extension of some $n$-bit Boolean function. Then, $ \stabs^{(b)}(f) = \slys^{(b)}(f) = \sen_\un^{(b)}(f)$ for $b \in \{0, 1\}$ and $\slys^{(u)}(f) \le \sen_\un^{(u)}(f)$
\end{claim}
\begin{proof}
    It is easy to see that for any $x \in \zuon, b \in \{0, 1\}$ such that $f(x) = b, \stabs(f, x) = \slys(f, x) \le \sen_\un(f, x)$ and for any input $x$ such that $f(x)  = \un, \slys(f, x) \le \sen_\un(f, x)$.
    From ~\ref{thm:senstr} we know that $\sen_\un^{(1)}(f)$ occurs at the input $x$ that corresponds to the largest sized prime implicant (or implicate for $\sen_\un^{(0)}$) or the minimum sized maximal monochromatic subcube where all the stable bits are sensitive. Hence at this input $\stabs^{(1)}(f, x) = \sen_\un^{(1)}(f, x) = \sen_\un^{(1)}(f)$ and by a similar argument for $\stabs^{(0)}$ we have $ \stabs^{(b)}(f) = \slys^{(b)}(f) = \sen_\un^{(b)}(f)$ for $b \in \{0, 1\}$
\end{proof}
Since we already know that $\sen_\un^{(u)}(f) \le \sen_\un^{(0)}(f) + \sen_\un^{(1)}(f) - 1$, we know that $\frac{\sen_\un(f)}{2} \le \stabs(f) \le \slys(f) \le \sen_\un(f)$ thus proving that these notions of sensitivity are at most a factor of 2 away from our choice of sensitivity and thus they would all satisfy the Sensitivity Conjecture.

\end{document}